\documentclass[11pt,a4paper]{article}
\usepackage{amsmath, amssymb, amsfonts, amsthm}
\usepackage{hyperref}

\newtheorem{theorem}{Theorem}[section]
\newtheorem{prop}[theorem]{Proposition}
\newtheorem{cor}[theorem]{Corollary}
\newtheorem{lemma}[theorem]{Lemma}

\newtheorem{rem}[theorem]{Remark}

\def\z*{\bar z}

\def\RE{\mathbb R}
\def\NN{\mathbb N}

\def\p{\par\noindent}
\newcommand{\lf}{\left}
\newcommand{\ri}{\right}

\newcommand{\be}{\begin{equation}}
\newcommand{\ee}{\end{equation}}
\newcommand{\bey}{\begin{eqnarray}}
\newcommand{\eey}{\end{eqnarray}}
\newcommand{\ls}{\leqslant}

\newcommand{\bees}{{\begin{split}}}
\newcommand{\ees}{{\end{split}}}
\newcommand{\bes}{{\begin{equation}\begin{split}}}
\newcommand{\es}{{\end{split}\end{equation}}}
\newcommand{\erre}{{\mathbb R}}

\newcommand{\sech} {{\rm{sech}}}

\newcommand{\bea}{\begin{eqnarray}}

\newcommand{\eea}{\end{eqnarray}}

\newcommand{\n}{\noindent}

\newcommand{\f}{\frac}

\newcommand{\al}{\alpha}

\newcommand{\si}{\sigma}

\newcommand{\mc}{\mathcal}

\numberwithin{equation}{section}

\newcommand{\wt}{\widetilde}

\newcommand{\donothing}[1]{}

\newcommand{\de}{\delta}

\newcommand{\GG}{\mathcal{G}}
\newcommand{\LG}[1]{L^{#1}}
\newcommand{\LT}[2]{L^{#1}_{[#2]}}
\newcommand{\phisol}[1]{\phi_{#1}}
\renewcommand{\Im}{\operatorname{Im}\,}
\renewcommand{\Re}{\operatorname{Re}\,}

\date{}
\usepackage{graphicx}
\begin{document}
\title{Fast solitons on star graphs}

\author{Riccardo Adami${}^1$, Claudio Cacciapuoti${}^2$,
Domenico Finco${}^3$, and Diego Noja${}^1$
\\
\\
${}^1$Dipartimento di Matematica e Applicazioni, Universit\`a
 di Milano
Bicocca \\
via R. Cozzi 53, 20125 Milano, Italy\\
riccardo.adami@unimib.it, diego.noja@unimib.it
 \\ \\
${}^2$Hausdorff Center for Mathematics,
 Institut f\"ur Angewandte Mathematik \\
Endenicher Allee 60, 53115 Bonn, Germany\\
cacciapuoti@him.uni-bonn.de\\ \\
${}^3$
Facolt\`a di Ingegneria, Universit\`a Telematica Internazionale Uninettuno\\
Corso Vittorio Emanuele II, 00186 Roma, Italy\\
d.finco@uninettunouniversity.net}

\maketitle

\begin{abstract}
\par\noindent
We define the Schr\"odinger equation with focusing, cubic
nonlinearity on one-vertex graphs. We prove global well-posedness in
the energy domain and conservation laws for some self-adjoint boundary
conditions at the vertex, i.e. Kirchhoff boundary condition and the so
called $\delta$ and $\delta'$ boundary conditions. Moreover, in the
same setting we study the collision of a fast solitary wave with the
vertex and we show that it splits in reflected and transmitted
components. The outgoing waves preserve a soliton character over a
time which depends on the logarithm of the velocity of the ingoing
solitary wave. Over the same timescale the reflection and transmission
coefficients of the outgoing waves coincide with the corresponding
coefficients of the linear problem. In the analysis of the problem we
follow ideas borrowed from the seminal paper \cite {[HMZ07]} about
scattering of fast solitons by a delta interaction on the line, by
Holmer, Marzuola and Zworski; the present paper represents
an extension of their work to the case of graphs and, as a byproduct,
it shows how to extend the analysis of soliton scattering by other
point interactions on the line, interpreted as a degenerate graph.
\end{abstract}

\begin{small}
\n
\emph{Keywords: }quantum graphs, non-linear Schr\"odinger equation, solitary waves.\\
\emph{MSC 2010: }35Q55, 81Q35, 37K40.
\end{small}

\section{Introduction}
\setcounter{equation}{0} In the present paper we study the {\it
  nonlinear} wave propagation on graphs.  As far as we know the
subject of nonlinear Schr\"odinger evolution on graphs is at its
beginnings. An extensive literature on the behaviour of linear wave
and Schr\"odinger equations on graphs exists (\cite
{[Kuc04],[Kuc05],[KS99],[BCFK06],[BEH]}, and references therein) and a
certain activity concerning the so called {\it discrete} nonlinear
Schr\"odinger equation (DNLSE) in chains with edges and graphs
inserted (``decorations'', interpreted as defects in the chain), both
from the physical and the numerical point of view (see for example
\cite {[KFTK],[BCSTV]}). To our knowledge, however, there are only
very few papers in which nonlinear Schr\"odinger evolution on graphs
has been introduced and studied (see \cite{[brando],[Sob]}). In the first paper \cite{[brando]}
(and similar ones quoted therein) NLS on graphs emerges in some models of quantum field theory on ``bulks'';
the second recent paper \cite {[Sob]} addresses from a physical point of view some
general questions related to the ones here studied and is briefly commented in
the conclusions of the present paper. Besides several results on existence and
stability of stationary states for the NLS on star graphs can be found in
\cite{[ACFN2]}.
\par In any case the study of nonlinear propagation in ramified
structures could be of relevance in several branches of pure and
applied science, from condensed matter physics to nonlinear fiber
optics, hydrodynamics and fluid transport (a non-traditional example
is blood flow in veins and arterias), and finally neural networks (see
for example the study of {\it reaction-diffusion type}
FitzHugh-Nagumo-Rall equations on networks in \cite{[CMu]}, and
references therein).\par In all these examples there is a strong
dependence on the modelization. The nonlinear Schr\"odinger equation
with cubic nonlinearity is especially suitable to describe nonlinear
electromagnetic pulse propagation in optical fibers and, under the
name of Gross-Pitaevskii equation, the dynamics of Bose-Einstein
condensates.  Better suited for other applications, for example the
hydrodynamic flow, is the KdV equation or its relatives, not treated
here. \par Of course real networks are
 not strictly
one dimensional, and an abstract graph, which is just a set of copies
of $\RE^+$ (``edges'', or ``branches'')
with functions living on it satisfying certain boundary conditions at $0$ lacks some of the geometric meaningful characteristics of a real
network, such as thickness and curvature of the branch or orientation between
edges. On the other hand,
problems
related to wave propagation on networks
are far from being well
understood also for the linear propagation (see for example \cite {[ACF],[CE]}),
and so we content ourselves with posing and analysing the nonlinear problem in the idealized and simplified
framework of an abstract graph.
\par We study here the special case of a star graph with three edges. A
generalization to a star graph with $n$ edges would be
straightforward, but here our interest is in clarifying the main
features of the evolution and the techniques involved in its analysis. A
preliminary
point and not a trivial issue is the definition of the dynamics.
Let us recall that for a star graph $\GG$, the linear Schr\"odinger
dynamics is defined by giving a self-adjoint operator $H$ on the
product of $n$ copies of $L^2(\RE^+)$ (briefly $L^2(\GG)$), with a
domain $D(H)$ in which appears a linear condition involving the values at $0$ of
the functions of the domain and of their derivatives. The admissible
boundary conditions characterize the
interaction at the vertex of the star graph. For the nonlinear
problem, we establish the well-posedness of the dynamics in the case
of a star graph with some distinguished boundary condition at the
vertex, namely the free (or Kirchhoff), and the $\delta$ and $\delta'$
boundary conditions (see section $2$ for the relevant definitions). \par
To
clarify
the exact problem we are faced to, the
differential equation to be studied is of the form
\begin{equation}
\label{diffform}
 i \frac{d}{dt}\Psi_t \ = \ H \Psi_t - | \Psi_t |^2 \Psi_t\,,\qquad t\geqslant0\,,
\end{equation}
where the function $\Psi$ is, for a three-edge graph, a column
vector
\begin{equation*}
\Psi =
\lf(
\begin{array}{c}
\psi_1 \\ \psi_2 \\ \psi_3
\end{array}
\ri)\ ,
\end{equation*}
that lies in the domain $D(H)$ of the linear Hamiltonian $H$ on the graph, so
emboding the relevant boundary conditions. This is the abstract strong form of
the equation, which is equivalent to a particular nonlinear
coupled system of scalar equations. The coupling is not due to the
nonlinearity, because of the definition
\begin{equation*}
|\Psi|^2\Psi\equiv \lf(
\begin{array}{c}
|\psi_1|^2\psi_1 \\ |\psi_2|^2\psi_2 \\ |\psi_3|^2\psi_3
\end{array}
\ri)\ ,
\end{equation*} but to the boundary condition at the vertex. For example,
in the simple case of a Kirchhoff boundary condition, the coupling between the edges is given by
\begin{equation*}
 \Psi \in L^2(\GG) \text{ s.t. } \psi_i \in H^2 (\RE^+ ), \, \psi_1
 (0) = \psi_2 (0) = \psi_3 (0), \, \psi_1' (0) + \psi_2' (0) +
 \psi_3'(0)=0 \ .
\end{equation*}
The $\delta$ or $\delta'$ boundary conditions allow a coupling between
the values of the function $\Psi$ and the values of their derivatives at
the origin, but in principle the nature of the problem is
unaltered.\par\noindent In the present paper, for several reasons, we
prefer to write the dynamics in weak form, which is the following
\begin{equation}
\label{intform1}
\Psi_t \ = \ e^{-iH t} \Psi_0 + i \int_0^t e^{-i
  H (t-s)}|\Psi_s|^2 \Psi_s \, ds\,,\qquad t\geqslant0\ ;
\end{equation}
here the $\Psi$
belongs to
the
form domain ${\mathcal D} ({\mathcal E}_{lin})$, where ${\mathcal
  E}_{lin}$ is the quadratic form of the linear Hamiltonian $H$. We
interpret, according to the use, the form domain as the finite
energy space.
An adaptation of the methods in
\cite{[AN09]} gives the local well-posedness of the equation
\eqref{intform1} for every initial data
in ${\mathcal D} ({\mathcal
  E}_{lin})$. Moreover, charge and energy conservation laws hold true
for such weak solutions, and as a consequence, the NLS on graph admits
global solutions. The generalization of the well-posedness
  result to the case of general self-adjoint boundary conditions will
  be treated elsewhere.

\par Apart
from well-posedness, the main goal of this paper is to provide
information on the interaction between a solitary wave and the
boundary condition at the vertex. As it is well-known, the NLS on the
line admits a family of solitary non dispersive solutions, rigidly
translating with a fixed velocity. A rich family of solitary solutions
(``solitons'') is given by the action of the Galilei group on the
function
\begin{equation*}
 \phi (x) \ = \ \sqrt 2 \cosh^{-1} x\ ,\qquad x\in\RE\,,
 \end{equation*}
 or
explicitly,
\begin{equation*}
 \phisol{x_0,v}(x,t) = e^{i \f v 2 x} {e^{- i t \f
    {v^2} 4}}e^{it} \phi ( x - x_0 - vt)\,\qquad x\in\RE\ ,\ t\in\RE\ ,\ v\in\RE.
\end{equation*}
\par

We show that, after
the collision of a
solitary wave with the vertex there exists
a timescale during which the dynamics can be described as the
scattering of three split solitary waves, one reflected on the same
branch where the originary soliton was running asymptotically in the
past, and two transmitted solitary waves on the other branches. On
the same timescale, the amplitudes of the reflected and transmitted
solitary waves are given by the scattering matrix of the linear
dynamics on the graph. The soliton-like character persists over time
intervals that depend on the velocity of the impinging original
soliton: the faster is the original soliton, the (logarithmically in
the velocity $v$) longer is the survival time of the solitary wave
behaviour on every branch of the graph. The non-trivial point is that
the persistence time of solitary behaviour after collision with the
vertex is much longer, for fast solitons, than the time over which it
is reasonable to approximate the nonlinear dynamics with the linear
one.
The same timescale of the order $\ln v$ of persistence of
solitary behaviour appears in the paper \cite{[AbFS]} where the
collision of two solitary waves with an underlying smooth potential is
studied, and in the paper by Holmer, Marzuola and Zworski
\cite{[HMZ07]} on the fast NLS-soliton scattering by a delta potential
on the line, which is the main source of inspiration for
 our result and for the techniques employed in the present
paper. \par We give now an outline of our result and its proof.\p The
initial data are of the following form
\begin{equation} \label{data}
 \Psi_0 (x) \ =
\ \left( \begin{array}{c} \chi (x) e^{-i \f v 2 x} \phi (x - x_0) \\ 0
  \\ 0 \end{array} \right), \qquad x_0 \geqslant v^{1 - \delta}, \ 0 < \delta < 1,
\end{equation}
 where $\chi$ is a cut off function,
that is $\chi \in C^\infty (\erre^+ )$, $\chi = 1$ in $(2, + \infty)$
and $\chi = 0$ in $(0,1)$. Apart from a small tail term truncated by
the cutoff function, the first component is the initial condition
of a
free (i.e. without external potentials) NLS which on the line yields
a solitary wave running with velocity $v$; the center $x_0$ of the
initial soliton is chosen far from the vertex.
We are interested in the evolution $\Psi_t$ of this
initial condition. \par
The
dynamics can be divided into three
phases. The pre-interaction phase, where the evolved initial condition is
far from the vertex, and the undisturbed NLS evolution dominates. At
the end of this phase, the solution enters the vertex zone, and differs
(in $L^2$ norm) from the evolved solitary wave by an exponentially
small error in the velocity $v$. The second phase is the interaction
phase, in which a substantial fraction of the mass of the initial
soliton has reached the vertex, and the linear dynamics dominates due
to the shortness of the interaction time, leaving the system at the
end of this phase with three scattered waves, the amplitudes of which
are given by the action of the scattering matrix of the associated
linear graph on the incoming solitary wave. The size of the
corresponding error is (again in $L^2$ norm) a suitable negative
inverse power ($v^{{-\frac{1}{2}}\delta}$) of the velocity. In
the phase one, the main technical tool consists in the accurate
use, as fixed by \cite{[HMZ07]}, of the Strichartz's estimates to control the deviations between
the unperturbed NLS flow and the NLS flow on the graph. In the phase
two, we need to compare the nonlinear evolution with the linear flow on
the graph in the relevant time interval. Finally there is the
post-interaction phase, where the free NLS dynamics dominates again;
however, now the initial data are not exact solitary waves, but waves
with soliton-like profiles and ``wrong'' amplitudes (due to the
scattering process in the interaction phase). \par The true evolution is compared
with a reference dynamics given by the superposition of the
nonlinear evolution of the outgoing scattered profiles, and it turns
out that the error is, in $L^2$ norm, of the order of an inverse power
of velocity (depending on the size of the time interval of
approximation). For a precise formulation one has to tackle the
problem of representing the reference soliton dynamics to be
compared with the true dynamics. This problem arises because one
would like to use crucial and known properties of NLS on the line
(such as existence of an infinite number of constants of motion), and various
associated estimates, while on a star graph one has a NLS on
halflines, jointly with boundary conditions. The problem occurs, of
course, in each of the three phases in which the dynamics is
decomposed.
\par
Our choice of reference dynamics is the following. We associate to
every edge of the star graph a companion edge chosen between the other
two, in such a way to have three fictitious lines; then we glue the
soliton on every single edge with the right tail on the companion
edge, respecting the free nonlinear dynamics. One of the main technical points in the
analysis of the true dynamics is to have a control in the errors brought by this schematization.
More precisely, let us define

\begin{equation*}
\tilde\Phi^1_t(x_1 ,x_2, x_3)\equiv
\begin{pmatrix}
\tilde r e^{-i \f {v^2} 4 t}e^{i\f v 2 x_1} e^{it} \phi (x_1 + x_0 - vt)\\
\tilde r e^{-i \f {v^2} 4 t}e^{-i\f v 2 x_2} e^{it} \phi (x_2 - x_0 + vt)\\
0
\end{pmatrix}
\end{equation*}
\begin{equation} \label{fittilde}
\tilde\Phi^2_{t}(x_1 ,x_2, x_3)\equiv
\begin{pmatrix}
0\\
\tilde t e^{-i \f {v^2} 4 t}e^{i\f v 2 x_2} e^{it} \phi (x_2 + x_0 - vt)\\
\tilde t e^{-i \f {v^2} 4 t}e^{-i\f v 2 x_3} e^{it} \phi (x_3 - x_0 + vt)
\end{pmatrix}
\end{equation}
\begin{equation*}
\tilde\Phi^3_{t}(x_1 ,x_2, x_3)\equiv
\begin{pmatrix}
\tilde te^{-i \f {v^2} 4 t}e^{-i\f v 2 x_1} e^{it} \phi (x_1 - x_0 + vt)\\
0\\
\tilde te^{-i \f {v^2} 4 t}e^{i\f v 2 x_3} e^{it} \phi (x_3 + x_0 - vt)
\end{pmatrix}
\end{equation*} Each of these vectors represents a soliton on the fictitious line
given by an edge and its companion, multiplied by the scattering
coefficients of the linear dynamics considered, Kirchhoff, $\delta$ or
$\delta'$ (and here left unspecified). Up to a small error, these functions
represent outgoing waves at the end ($t=t_2$) of the interaction phase,
which is essentially a scattering process. Taking these as initial
data for the free nonlinear dynamics on the pertinent fictitious line,
we define their time evolution $\Phi_t^{j}$ as given by
\begin{equation*}
\Phi^j_{t}=e^{-i H_j (t-t_2)} \tilde\Phi^j_{t_2} + i \int_{t_2}^t ds
\, e^{-i H_j (t-s)} | \Phi^j_s |^2 \Phi^j_s\qquad j=1,2,3,\,t\geqslant t_2\,,
\end{equation*}
where
the $H_j$ are the linear Hamiltonians that decouple the $j+2$-branch from the others.
With these premises, the main result of the paper is the following.

\begin{theorem}
Let $\Psi_t$ the unique, global solution to the Cauchy problem
\eqref{diffform} with initial
data \eqref{data}.
There exist $\tau_* >0$ and $T_{\ast}>0 $ such that for $0< T <T_{\ast}$ one has
\begin{equation*}
\| \Psi_t - \Phi^1_t- \Phi^2_t- \Phi^3_t \|_{L_x^2(\GG)} \ \leqslant \ C v^{-\f{T_{\ast}-T}{\tau_*}}
\end{equation*}
for every time $t$ in the interval $t_2 < t < t_2 + T\ln v\ ,$
where $C$ is a positive constant independent of $t$ and $v$.
\end{theorem}
To be precise, the Hamiltonians in \eqref{diffform} to which the
theorem refers have to be rescaled in order to give a nontrivial
scattering matrix in the regime of high velocity (see section 4, in
particular theorem \ref{mainth}).

To get the previous result, Strichartz estimates do not suffice, and
more direct properties coming from the integrable character of the
cubic NLS are needed. In particular, thanks to the existence of an
infinite number of integrals of motion, in \cite{[HMZ07]} a spatial
localization property of the solution of NLS with smooth data is
proven, with a polynomial (and not exponential) bound in time. This
gives a control on the tails of the difference between the solution
and the reference modified solitary dynamics. An analogous method
applies in our case. Let us note that as a consequence of the previous result, we can give the scattered amplitudes in
terms of the incoming amplitude and scattering coefficients of the
linear dynamics, in the time range of applicability of the main theorem (see remark 4.4).
\p We give a brief summary of the content of the
various sections of the papers. In section 2 we give some generalities
on linear dynamics on graphs, including Hamiltonians, their quadratic
forms, resolvents and propagators. Moreover the essential Strichartz
estimates are recalled.\p In section 3 local and global well posedness
of nonlinear Schr\"odinger equations on star graphs is proved.\p In
section 4 the main result (theorem \ref{mainth}) is introduced and
stated.\p Section 5 is devoted to the proof of the result.
In section 6 some final remarks are given and further possible
developments are discussed.

\subsection{Setting and notations}
We consider a graph $\GG$ given by three infinite half lines attached
to a common vertex.
In order to study a quantum mechanical problem on $\GG$,
the natural Hilbert space is then
$L^2(\GG)\equiv L^2(\RE^+)\oplus L^2(\RE^+)\oplus L^2(\RE^+)$.

\n
We denote the elements of $L^2(\GG)$ by capital greek letters,
while functions in $L^2(\RE^+)$ are
 denoted by lowercase greek letters.
It is convenient to represent functions in $L^2(\GG)$ as column vectors
of functions in $L^2(\RE^+)$, namely
\begin{equation*}
\Psi =
\lf(
\begin{array}{c}
\psi_1 \\ \psi_2 \\ \psi_3
\end{array}
\ri).
\end{equation*}
The norm of $L^2$-functions on $\GG$ is naturally defined by
$$
\| \Psi \|_{L^2 (\GG)} : = \left( \sum_{j=1}^3 \| \psi_j \|^2_{L^2
  (\erre^+)} \right)^{\f 1 2}.
$$
Analogously, given $1 \leqslant r \leqslant \infty$,
we define the space $L^r (\GG)$ as the set of functions
on the graph whose components are elements of the space $L^r (\erre^+)$,
and the norm is correspondingly defined by
\begin{equation*}
\big\|\Psi\big\|_{\LG{r}
  (\GG)}=\bigg(\sum_{j=1}^3\|\psi_j\|_{L^r(\RE^+)}^{r}\bigg)^{\f 1 r},
\ 1 \leqslant r < \infty, \qquad \big\|\Psi\big\|_{\LG{\infty}
  (\GG)}= \sup_{1 \leqslant j \leqslant 3}\|\psi_j\|_{L^\infty(\RE^+)} .
\end{equation*}
When a functional norm refers to a function defined on the graph,
we omit the symbol $\GG$. Furthermore, from now on, when such a norm is $L^2$, we
drop the subscript, and simply write $\| \cdot \|$. Accordingly, we denote by $(\cdot,\cdot)$ the scalar product in $L^2$.

As it is standard when dealing with Strichartz's estimates, we make
use of spaces of functions that are measurable as functions of both
time (on the interval $[T_1,T_2]$) and space (on the graph). We denote such spaces by $L^p_{[T_1,T_2]}{\LG{r} (\GG)}$,
with indices
$1\leqslant r\leqslant \infty$,
$1\leqslant p < \infty$; we
endow them with the norm
\begin{equation*}
\big\|\Psi\big\|_{\LT{p}{T_1,T_2}\LG{r}{(\GG)}}=
\bigg(\int_{T_1}^{T_2}\big\|\Psi_s\big\|_{\LG{r}}^p
ds\bigg)^{1/p}, \ 1 \leqslant p < \infty, \qquad
\big\|\Psi\big\|_{\LT{\infty}{T_1,T_2}\LG{r}{(\GG)}}=\sup_{t \in
  [T_1,T_2]} \|\Psi_s\big\|_{\LG{r}}.
\end{equation*}
\n
The extension of the definitions given above to the case $p=\infty$ 
$r=\infty$
is
straightforward.

\n
Besides, we need to introduce the spaces
$$H^1(\GG) \equiv H^1(\erre^+) \oplus H^1(\erre^+) \oplus
H^1(\erre^+), \qquad
H^2(\GG) \equiv H^2(\erre^+) \oplus H^2(\erre^+) \oplus
H^2(\erre^+), $$
equipped with the norms
\be \label{sobbo}
\| \Psi \|_{H^1(\GG)}^2 \ = \ \sum_{i=1}^3 \| \psi_i \|_{H^1(\erre^+)}^2,
\qquad
\| \Psi \|_{H^2(\GG)}^2 \ = \ \sum_{i=1}^3 \| \psi_i \|_{H^2(\erre^+)}^2.
\ee

\n
The product of functions is defined componentwise,
\begin{equation*}
\Psi\Phi\equiv
\lf(
\begin{array}{c}
\psi_1\phi_1 \\ \psi_2\phi_2 \\ \psi_3\phi_3
\end{array}
\ri),
\qquad \text{so that} \qquad
|\Psi|^2\Psi\equiv
\lf(
\begin{array}{c}
|\psi_1|^2\psi_1 \\ |\psi_2|^2\psi_2 \\ |\psi_3|^2\psi_3
\end{array}
\ri).
\end{equation*}

\n
We denote by $\mathbb I$ the $3\times 3$ identity matrix, while $\mathbb J$
is the $3\times 3$ matrix whose elements are all equal to one.

\n
When an element of $L^2 (\GG)$ evolves in time, we use in notation the
subscript $t$: for instance, $\Psi_t$. Sometimes we shall write $\Psi(t)$ in order
to highlight the dependence on time, or whenever such a notation is more
understandable.

\section{\label{sec:summary}Summary on linear dynamics on graphs}

\subsection{\label{subsec2.1}Hamiltonians and quadratic forms}
Standard references about the linear Schr\"odinger equation on graphs are \cite{[BCFK06],[BEH],[Kuc04],[Kuc05],[KS99]}, to which we refer for more extensive treatments.
Here we only give the definitions needed to have a self-contained exposition.
\par
We consider three Hamiltonian operators, denoted by $H_F$,
$H_\delta^\alpha$, $H_{\delta^\prime}^\beta$ (with $\alpha,\beta\in\RE$), and called, respectively, the
Kirchhoff, the Dirac's delta, and the delta-prime Hamiltonian.
These operators act as
\begin{equation}
\Psi \ \longmapsto \
\lf(
\begin{array}{c}
-\psi_1'' \\ -\psi_2'' \\ -\psi_3''
\end{array}
\ri)
\label{san}
\end{equation}
on some subspace of $H^2 (\mathcal G)$, to be defined by suitable
boundary conditions at the vertex.

Here and in the following subsection
 we collect some basic facts (see
\cite{[KS99]},
\cite{[Kuc04]},\cite{[Kuc05]} \cite{[BCFK06]})
on $H_F$, $H_{\delta}^\alpha$, and $H_{\delta^\prime}^\beta$.

The Kirchhoff Hamiltonian $H_F$ acts on the domain
\begin{equation}
{\mathcal D} (H_F): = \{ \Psi \in H^2(\GG) \text{ s.t. } \,
\psi_1 (0) = \psi_2 (0) = \psi_3 (0), \, \psi_1' (0) + \psi_2' (0)
+ \psi_3'(0)=0 \}.
\label{fuji}
\end{equation}
It is well-known, see \cite{[KS99]}, that \eqref{fuji}
and \eqref{san} define a self-adjoint Hamiltonian on $L^2(\GG)$.
Boundary conditions in \eqref{fuji} are usually called Kirchhoff
boundary conditions. We use the index $F$
to remind that $H_F$ reduces to the free Hamiltonian on the line for
a degenerate graph composed of two half lines.

\n
The quadratic form ${\mathcal E}_F$ associated to $H_F$ is defined on
the subspace
\begin{equation*}
{\mathcal D} ({\mathcal E}_F) = \{ \Psi \in H^1(\GG) \text{ s.t. } \,
\psi_1 (0) = \psi_2 (0) = \psi_3 (0) \}
\end{equation*}
and reads
\begin{equation*}
{\mathcal E}_F [\Psi ] = \sum_{i=1}^3 \int_0^{+\infty} |\psi_i ' (x) |^2 \,dx\,.
\end{equation*}

The Dirac's delta Hamiltonian is defined on
the domain
\begin{equation} \label{domdelta}
{\mathcal D} (H_\delta^\alpha): =
\{ \Psi \in H^2(\GG) \text{ s.t. } \,
\psi_1 (0) = \psi_2 (0) = \psi_3 (0), \, \psi_1' (0) + \psi_2' (0) +
\psi_3'(0)= \alpha \psi_1 (0) \}
\end{equation}
Again, $H_\delta^\alpha$ is a self-adjoint opeator on $L^2(\GG)$
(\cite{[KS99]}).
It appears that $H_\delta^\alpha$ generalizes the ordinary Dirac's delta interaction with strength parameter $\alpha$ on the line, see, e.g. \cite{aghh:05}.

\n
The quadratic form ${\mathcal E}_\delta^\alpha$
associated to $H_\delta^\alpha$ is defined on
\begin{equation*}
{\mathcal D} ({\mathcal E}_\delta^\alpha) = \{ \Psi \in H^1 (\GG)
\text{ s.t. } \,
\psi_1 (0) = \psi_2 (0) = \psi_3 (0) \}
\end{equation*}
and is given by
\begin{equation*}
{\mathcal E}_\delta^\alpha [\Psi ] = \sum_{i=1}^3 \int_0^{+\infty}
|\psi_i ' (x) |^2 \,dx\, + \alpha | \psi_1 (0) |^2.
\end{equation*}

The delta-prime Hamiltonian is defined on the domain
\begin{equation}
{\mathcal D} (H_{\delta^\prime}^\beta): =
\{ \Psi \in H^2(\GG) \text{ s.t. } \,
\psi_1^\prime (0) = \psi_2^\prime (0) = \psi_3^\prime (0),
\, \psi_1 (0) + \psi_2 (0) +
\psi_3(0)= \beta \psi_1^\prime (0) \}
\label{domprime}
\end{equation}
Again, $H_{\delta^\prime}^\beta$ is a self-adjoint opeator on $L^2(\GG)$
(\cite{[KS99]}).

\n
The quadratic form ${\mathcal E}_{\delta^\prime}^\beta$
associated to $H_{\delta^\prime}^\beta$ is defined on
$
{\mathcal D} ({\mathcal E}_{\delta^\prime}^\beta) \ = \ H^1 (\GG)
$
and is given by
\begin{equation*}
{\mathcal E}_{\delta^\prime}^\beta [\Psi ] = \sum_{i=1}^3 \int_0^{+\infty}
|\psi_i ' (x) |^2 \,dx\, + \frac{1}{\beta} \left| \sum_{i=1}^3 \psi_i (0) \right|^2.
\end{equation*}
\p
Notice that $H_{\delta^\prime}^\beta$ does not reduce to the standard
$\delta'$ interaction on the line, see, e.g., \cite{aghh:05}, when it
is restricted to a two-edge graph. Here we are following the notation
in \cite{[Kuc04]}, \cite{[Kuc05]}. The present $\delta'$ vertex is
called sometimes $\delta'_s$ graph, where $s$ is for symmetric. A
discussion of the correct extension of the usual $\delta'$
interaction is
given in \cite{[Ex]} and \cite{[BEH]}. For completeness, we give
the operator domain (the action is the same as in the other cases).
We use the denomination $\tilde\delta'$ to avoid confusion with the previously defined interaction.
\begin{equation*}
{\mathcal D} (H_{\tilde\delta^\prime}^{\beta}): =
\bigg\{ \Psi \in H^2(\GG) \text{ s.t. } \, \sum_{j=1}^n \psi_j^\prime(0)=0 \
,\quad
\psi_j(0) -\psi_k (0) =\frac{\beta}{n}(\psi_j^\prime(0) -\psi_k^\prime(0))\
,\quad j,k= 1,2,...,n\bigg\} \ .
\end{equation*}

Throughout the paper we restrict to the case of repulsive
delta and delta-prime ($\delta'$) interaction, i.e., $\alpha, \beta > 0$. It
is easily proved, for example by inspection of the operator resolvents in the following subsection, that such a condition prevents the corresponding
Hamiltonian operator from possessing bound states.

We also point out that the Hamiltonian $H_\delta^\alpha$ is well
defined for
$\alpha=0$ and that indeed $ H_\delta^\alpha\big|_{\alpha=0}\equiv H_F$.
 We finally point out that, fixed $\tilde \alpha, \tilde \beta > 0$,
  we shall consider the Hamiltonian operators $H_\delta^{\tilde \alpha
    v}$ and $H_{\delta^\prime}^{\tilde \beta /
    v}$, namely, in the following we rescale $\alpha = \tilde \alpha
v$ and $\beta = \tilde \beta / v$.

\subsection{Resolvents, propagators and scattering}
For any complex number $k$ with $\Im k>0$ we denote
by
$R_F (k), R_\delta^\alpha (k), R_{\delta^\prime}^\beta$ the resolvents
$(H_F - k^2)^{-1}, (H_\delta^\alpha - k^2)^{-1},
(H_{\delta^\prime}^\beta- k^2)^{-1}$, respectively.

\n
We define the function $U_t$ by
\begin{equation*}
U_t (x) \ : = \ \f {e^{i \f {x^2}{4t}}} {\sqrt{4 \pi i t}}\ , \qquad t\neq 0.
\end{equation*}
In the following we shall use the same symbol $U_t$
to denote the operator in $L^2(\RE)$ defined by
\begin{equation*}
\big[U_t\psi\big](x):=\int_{-\infty}^\infty U_t(x-y)\psi(y)dy\ ,\quad t\neq 0\,.
\end{equation*}
Moreover, we define two integral operators $U_t^\pm$ acting on $L^2 (\erre^+)$
\begin{equation*}
U_t^\pm \ : \ L^2 (\erre^+) \rightarrow L^2 (\erre^+), \quad
\big[U_t^\pm \psi\big](x) = \int_0^{+ \infty} U_t ( x \pm y ) \psi (y) \, dy\ , \qquad t\neq 0\ .
\end{equation*}
We stress that, according to our definitions, the operators
 $U^-_t$ and $U_t$ have the same integral kernel,
 but act on different
Hilbert spaces.

For the cases we consider, resolvents and propagators can be easily
computed (see, e.g. \cite[pp. 201--226]{[BCFK06]} for resolvent formulas
with generic boundary conditions in the vertex). The results are summarized in the following theorem.

\begin{theorem} \label{dignita}
For any complex number $k$ with $\Im k>0$, the integral kernel of the resolvent operators
$R_F (k), R_\delta^\alpha (k), R_{\delta^\prime}^\beta(k)$ are given by
\begin{eqnarray}
R_F (k; x,y) & = & \f{i}{2k } e^{i k |x-y| } \mathbb{I} \label{resofree}
 +\f{i}{2k } e^{i k (x+y)} \f13
\lf(
\begin{array}{ccc}
-1 & 2 & 2 \\
2 & -1 & 2 \\
2 & 2 & -1
\end{array} \ri)\, , \\
R_{\delta}^\alpha (k; x, y) & = & \label{resodelta}
\f{i}{2k } e^{i k |x-y| } {\mathbb I}
 - \f{i}{2k } \f {e^{i k(x+y) }} {\al-3ik}
\lf(
\begin{array}{ccc}
{\al-ik} & {2ik} & {2ik} \\
{2ik} & {\al-ik} & {2ik} \\
{2ik} & {2ik} & {\al-ik}
\end{array} \ri)\,, \\
R_{\de '}^\beta (k ; x, y) & = & \label{resodeltaprima}
\f{i}{2k } e^{i k |x-y| } {\mathbb I}
 - \f{i}{2k } \f{e^{i k(x+y) }} {3-i\beta k}
\lf(
\begin{array}{ccc}
-1+i\beta k & {2} & {2} \\
 {2} & -1+i\beta k & {2} \\
 {2} & {2} & -1+i\beta k
\end{array}
\ri).
\end{eqnarray}
Furthermore, the unitary group of the related time evolution
operators reads
\begin{eqnarray}
\label{freepro}
e^{- i H_F t} & = & U_t^- \mathbb{I} + U_t^+ \f 1 3
\left( \begin{array}{ccc} -1 & 2 & 2 \\ 2 & -1 & 2 \\ 2 & 2 & -1
\end{array} \right)\, = \,
( U_t^- - U_t^+ ) {\mathbb I} + \f 2 3 U_t^+ {\mathbb J} \\
 \label{propdelta}
U_{\delta,t}^\alpha (x, y) & = & [ U_t (x-y) - U_t (x + y)] {\mathbb
  I} + \f 2 3 \left[ U_t (x + y) - \f \alpha 3 \int_0^{+ \infty} du \,
  e^{- \f \alpha 3 u} U_t (x + y + u)
\right] {\mathbb J}, \\
\label{propdeltaprime}
U_{\delta^\prime,t}^\beta (x, y) & = & [ U_t (x-y) + U_t (x + y)] {\mathbb
  I} - \f 2 \beta \int_0^{+ \infty} du \,
  e^{- \f 3 \beta u} U_t (x + y + u)
 {\mathbb J},
\end{eqnarray}
where
$x,y\in\RE^+$, $\alpha, \beta >0$.
\end{theorem}
\begin{proof}
We start from the proof of
\eqref{resodelta}. Let $(\Psi)_i = \psi_i$ for $i=1,2,3$ and define $R\Psi$ by
\begin{equation}
\lf( R\Psi \ri)_i (x) = \f{i}{2k} \int_0^{+\infty} e^{ik|x-y| } \psi_i(y)\, dy + \f{i}{2k}e^{ikx} c_i
\qquad i=1,2,3
\end{equation}
where $c_i=c_i(\Psi)$ are constant to be specified. It is obvious that $\lf(
R\Psi \ri)_i$ satisfies
\begin{equation}
\lf( -\f{d^2}{dx^2} -k^2 \ri) \lf( R\Psi \ri)_i=0 \qquad i=1,2,3
\end{equation}
then it is sufficient to fix $c_i$ such that $R\Psi$ belongs to ${\mathcal D}(H_{\al}^{\de}) $ in order
to compute $R_{\al}^{\de} $. The boundary conditions in \eqref{domdelta}
translate to the following linear system
for $c_i$.
\begin{equation}
\lf\{
\begin{aligned}
&\int_0^{+\infty} e^{iky } \psi_1(y)\, dy + c_1 = \int_0^{+\infty} e^{iky } \psi_2(y)\, dy + c_2 \\
&\int_0^{+\infty} e^{iky } \psi_2(y)\, dy + c_2 = \int_0^{+\infty} e^{iky } \psi_3(y)\, dy + c_3 \\
&\int_0^{+\infty} e^{iky } (\psi_1(y)+\psi_2(y)+\psi_3(y)) \, dy -( c_1 +c_2 + c_3) =
\f{i\al}{k} \lf(\int_0^{+\infty} e^{iky } \psi_2(y)\, dy + c_2 \ri)
\end{aligned}
\ri.
\end{equation}
the solution is easily computed and it is given by
\begin{equation}
\lf\{
\begin{aligned}
&c_1 =\int_0^{+\infty} e^{iky } \lf( -\f{k+i\al}{3k+i\al}\psi_1(y)+\f{2k}{3k+i\al}\psi_2(y)
+\f{2k}{3k+i\al}\psi_3(y) \ri)\, dy \\
&c_2 =\int_0^{+\infty} e^{iky } \lf( \f{2k}{3k+i\al}\psi_1(y)-\f{k+i\al}{3k+i\al}\psi_2(y)
+\f{2k}{3k+i\al}\psi_3(y) \ri)\, dy\\
&c_3 =\int_0^{+\infty} e^{iky } \lf( \f{2k}{3k+i\al}\psi_1(y)+\f{2k}{3k+i\al}\psi_2(y)
-\f{k+i\al}{3k+i\al}\psi_3(y) \ri)\, dy
\end{aligned}
\ri.
\end{equation}
which gives \eqref{resodelta}. In order to obtain \eqref{resofree} it is sufficient to put $\al=0$ in \eqref{resodelta}.
Formula \eqref{resodeltaprima} can be proved by the same method.

Now we prove \eqref{propdelta}.
We start from the standard formula 
$$
U_{\delta,t}^\alpha (x, y) \ = \ \f 1 {\pi i} \int_{- \infty}^{+ \infty}
R_{\delta}^\alpha (k; x, y) \, k d k,
$$
use \eqref{resodelta} and recall the following identity
\be \label{formul}
\f 1 {2 \pi} \int_{- \infty}^{+ \infty} \f {e^{i k (x + y)}} {a
  - i k} \, e^{- i k^2 t} \, dk \ =
\ \int_0^{+ \infty} e^{- a
  u} \, \f{e^{i \f{(x+y+u)^2}{4t}}}{\sqrt{4 \pi i t}} \, du,
\ee
then we immediately arrive at \eqref{propdelta}. Formula \eqref{freepro} can be
obtained
by putting $\al=0$ into \eqref{propdelta}. Formula \eqref{propdeltaprime} can be proved in the same
way.
\end{proof}
\begin{cor}
From the expression of the resolvent one immediately has the
 reflection and transmission coefficients:
\begin{eqnarray}
r (k) & = & - \f 1 3, \qquad t (k) \ = \ \f 2 3, \\
\label{scoeffd}
r_{H_\delta^\alpha} (k) & = & - \f {k + i \alpha} {3k + i \alpha}, \qquad
t_{H_\delta^\alpha} (k) \ = \ \f {2k} {3k + i \alpha}, \\
\label{scoeffdp}
r_{H_{\delta^\prime}^\beta} (k) & = & \f {\beta k + i} {\beta k + 3 i }, \qquad
t_{H_{\delta^\prime}^\beta} (k) \ = \ - \f {2i} {\beta k + 3 i }.
\end{eqnarray}
\end{cor}
We refer to \cite{[KS99]} for a comprehensive  analysis of the scattering
 on star graphs. Indeed by using
the results in \cite{[KS99]} the reflection and transmission coefficients can be
obtained directly by the boundary conditions in the vertex.

Throughout the paper we shall need some auxiliary dynamics to be compared
with the dynamics described by \eqref{diffform}, so,
for later convenience, we introduce the two-edge Hamiltonians $H_j $
and the corresponding two-edge propagators $e^{- i H_j t} $, $j=1,2,3$.

\n
Let $H_j $ be defined by:
\begin{equation*}
\begin{aligned}
{\mathcal D} (H_j) := \{ \Psi \in H^2(\GG) \text{ s.t. } \,
&
\psi_j(0)=\psi_{j+1}(0)\,,\,\psi_j'(0)+\psi_{j+1}'(0)=0\, ,
\psi_k(0)=0\, ,\,k\neq j,j+1 \}
\end{aligned}
\label{honshu}
\end{equation*}
\begin{equation*}
H_j \Psi :=
\lf(
\begin{array}{c}
-\psi_1'' \\ -\psi_2'' \\ -\psi_3''
\end{array}
\ri)\,,
\end{equation*}
where in equation \eqref{honshu} it is understood that $j=\{1,2,3\}$ modulo 3.

\n
The Hamiltonian $H_j$ couples the edges $j$ and $j+1$ with a Kirchhoff boundary condition and sets a Dirichlet
boundary condition for the remaining edge, so that there is free propagation between the edges $j$ and $j+1$
and no propagation between them and the edge $j+2$.
\begin{figure}[h!]
\begin{center}
\includegraphics[width=0.9\textwidth]{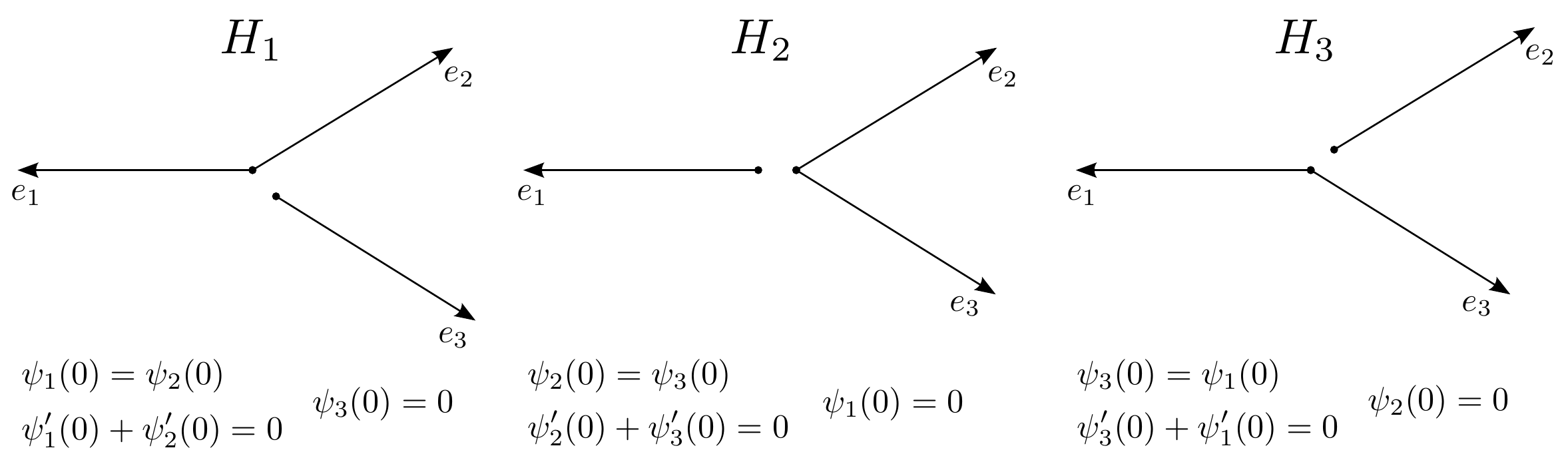}
\caption{\label{fig1}The open graphs depict the Hamiltonians $H_1$, $H_2$ and
$H_3$. Under each graph we report the boundary conditions on vectors in the
domain of the corresponding Hamiltonian.}
\end{center}
\end{figure}

\n
With a straightforward computation we have
\begin{equation}
\label{twoedgepro-j}
e^{- i H_j t} \ = \ U_t^- \mathbb{I} + U_t^+ {\mathbb T}_j\,,\qquad t\in\RE\,,
\end{equation}
where ${\mathbb T}_j$ are the matrices
\begin{equation*}
{\mathbb T}_1=\begin{pmatrix}
0&1&0\\
1&0&0\\
0&0&-1
\end{pmatrix}\;;\quad
{\mathbb T}_2=\begin{pmatrix}
-1&0&0\\
0&0&1\\
0&1&0
\end{pmatrix}\;;\quad
{\mathbb T}_3=\begin{pmatrix}
0&0&1\\
0&-1&0\\
1&0&0
\end{pmatrix}\,.
\end{equation*}

\subsection{Strichartz's estimates}
A key tool in our method is the extension of the
standard Strichartz's estimates (see e.g. \cite{[Caz]}) to
the dynamics on $\GG$ described by the propagators $e^{-i H_F t}$,
$e^{-i H_\delta^\alpha t}$, and $e^{-i H_{\delta^\prime}^\beta t}$.

\n
In this subsection we use the symbol $H$ to denote any of the three
Hamiltonians of interest described in section \ref{subsec2.1}.

As a preliminary step, we remark that from equations
\eqref{freepro}, \eqref{propdelta}, \eqref{propdeltaprime},
the standard
dispersive estimate immediately follows:
\begin{equation}
\big\|e^{-iH t}\Psi\big\|_{\LG{\infty}}\leqslant
\frac{c}{t^{1/2}}\big\|\Psi\big\|_{\LG{1}},\quad t\neq0\,.
\label{osaka}
\end{equation}

\begin{prop}[Strichartz Estimates for $e^{-iH t}$] \mbox{} \\
\label{prop:stric}
\n
Let $\Psi_0 \in L^2 (\GG)$, $\Gamma \in L^q_\erre \LG{k}$, with
$1\leqslant q,k \leqslant 2,
\frac{2}{q}+\frac{1}{k}=\frac{5}{2}$,
 and define
\begin{equation*}
\Psi(t) = e^{-i H t} \Psi_0 \qquad \Phi(t) = \int_0^t ds\, e^{-i H (t-s)} \Gamma(s)\,.
\end{equation*}
The following estimates hold true:
\begin{equation}
\label{stric1}
\lf\| \Psi \ri\|_{L^{p}_\erre \LG{r} } \leqslant c \| \Psi_0 \|
\end{equation}
\begin{equation}
\label{stric2}
\lf\| \Phi \ri\|_{L^p_\erre \LG{r} } \leqslant c \lf\| \Gamma
\ri\|_{L^q_\erre \LG{k} }
\end{equation}

\n
for any pair of indices $(r,p)$ satisfying
\begin{equation}
\label{admissible}
2\leqslant p,r\leqslant \infty \;,\; \frac{2}{p}+\frac{1}{r}=\frac{1}{2}.
\end{equation}
The constants
$c$ in \eqref{stric1} and \eqref{stric2} are independent of $T$.
\end{prop}
\begin{proof}
The proof is standard due to the dispersive estimate \eqref{osaka},
see for instance \cite{[Caz]} and \cite{[KT]}.
\end{proof}

\begin{rem} {\em If $H = H^\alpha_\delta$ ($H^{\beta}_{\delta^\prime}$),
the constants appearing in \eqref{osaka},
\eqref{stric1} and \eqref{stric2} are independent of $\alpha$ ($\beta$).
Indeed, by the change of variable $u \to u \alpha$ ($u \to u / \beta$)
 the integral term in \eqref{propdelta} (\eqref{propdeltaprime}) can be easily
estimated independently of $ \alpha$ ($\beta$), obtaining a dispersive estimate
\eqref{osaka} independent of the parameters and therefore, by the standard
Strichartz machinery, uniform inequalities \eqref{stric1} and \eqref{stric2}.
}
\end{rem}

\section{\label{sec:wp}Well-posedness and conservation laws}

Here we treat
the problem of the well-posedness (in the sense of, e.g., \cite{[Caz]}), i.e., the existence and uniqueness of
the solution to equation \eqref{intform1}
in the energy domain of the system. Such a domain turns out
to coincide with the form domain of the linear part of equation
\eqref{diffform}. Throughout this section, such a linear part is
denoted by $H$, and, according to the particular case under consideration,
it can be understood as the
Hamiltonian operator $H_F$, $H_\delta^\alpha$, or $H_{\delta^\prime}^\beta$.
Correspondingly, we denote the associated energy domain simply by
${\mathcal D} ({\mathcal E})$. All of the following formulas can be
specialized to the particular cases ${\mathcal D} ({\mathcal E}_F)$,
${\mathcal D} ({\mathcal E}_\delta^\alpha)$, or
${\mathcal D} ({\mathcal E}_{\delta^\prime}^\beta)$.

Let us stress that throughout the paper we do not approximate the dynamics
in $H^1$, but rather in $L^2$. Furthermore, local well-posedness in
$L^2$ is ensured by Strichartz estimates (proposition
\ref{prop:stric}), as is easily seen following the line
exposed in \cite{[Caz]}, chapters 2 and
3. Nonetheless, we prefer to deal with functions in the energy domain,
 since they are physically more
meaningful.

We follow the traditional line of proving, first of all, local
well-posedness,
and then extending it to all times by means of
a priori estimates provided by the
conservation
laws. 

\n
For a more extended treatment of the analogous problem for a
two-edge vertex (namely, the real line with a point interaction at the
origin), see \cite{[AN09]}.

First, we endow the energy domain ${\mathcal D}
({\mathcal E})$ with the $H^1$-norm defined in \eqref{sobbo}.
Second, we denote by ${\mathcal D}
({\mathcal E})^\star$ the dual of ${\mathcal D}
({\mathcal E})$, i.e., the
set of the continuous linear functionals on ${\mathcal D} ({\mathcal
  E})$. We denote the dual product of $\Gamma \in {\mathcal D}
({\mathcal E})^\star$ and $\Psi \in {\mathcal D} ({\mathcal E})$
by
$
\langle \Gamma, \Psi \rangle$. In such a bracket we sometimes exchange the
place of the factor in ${\mathcal D}
({\mathcal E})^\star$ with the place of the factor in ${\mathcal D}
({\mathcal E})$: indeed, the duality product follows the same
algebraic rules of the standard scalar product.

\n
As usual, one can extend the action of $H$ to the space
${\mathcal D} ({\mathcal E})$, with values in ${\mathcal D}
({\mathcal E})^\star$, by
$$
\langle H \Psi_1 , \Psi_2 \rangle \ : = \ ( H^{\f 1 2} \Psi_1,
 H^{\f 1 2} \Psi_2 ),
$$
where $(\cdot, \cdot)$ denotes the standard scalar product in
  $L^2 ({\mc G})$.

Furthermore, for any $\Psi \in {\mathcal D} ({\mathcal E})$
the identity
\begin{equation}
\label{extder}
\f d {dt} e^{-i H t} \Psi \ = \ - i H
e^{-i H t} \Psi
\end{equation}
 holds in ${\mathcal D} ({\mathcal E})^\star$
too. To prove it, one can first test the functional $\f d {dt} e^{-i
  H t} \Psi$ on an element $\Xi$ in the operator domain ${\mathcal
  D} (H)$, obtaining
\begin{equation*} \nonumber
\left\langle \f d {dt} e^{-i
  H t} \Psi, \Xi \right\rangle \ = \ \lim_{h \to 0} \left( \Psi, \f {e^{i
  H (t+h)} \Xi - e^{i H t} \Xi} h \right)\ = \ ( \Psi, i H e^{i H
t} \Xi) \ = \
\langle -i H e^{-i H t} \Psi, \Xi \rangle.
\end{equation*}
Then, the result can be extended to $\Xi \in {\mathcal D} ({\mathcal
  E})$ by a density argument.

\n
Besides, by \eqref{extder}, the differential version \eqref{diffform} of the
Schr\"odinger equation holds in ${\mathcal D} ({\mathcal
  E})^\star$.

In order to prove a well-posedness result we
need to generalize standard one-dimensional Gagliardo-Nirenberg
estimates to graphs, i.e.
\begin{equation}
\label{gajardo}
\| \Psi \|_{L^p} \ \leqslant \ C \| \Psi^\prime \|^{\f 1 2 - \f 1 p}_{L^2}
\| \Psi \|^{\f 1 2 + \f 1 p}_{L^2}\ \qquad +\infty \geqslant p \geqslant 1\ ,
\end{equation}
where the $C > 0$ is a positive constant which depends on the index $p$ only.
The proof of \eqref{gajardo}
follows immediately from the analogous estimates for functions of the
real line,
considering that any function in $H^1 (\erre^+)$ can be extended to an even function
in $H^1 (\erre)$, and applying this reasoning to each component of $\Psi$.

\begin{prop}[Local well-posedness in ${\mathcal D} ({\mathcal E})$]
\label{loch2}
For any $\Psi_0 \in {\mathcal D} ({\mathcal E})$, there exists $T > 0$ such that the
equation \eqref{intform1} has a unique solution $\Psi \in C^0 ([0,T),
{\mathcal D} ({\mathcal E}) )
\cap C^1 ([0,T), {\mathcal D} ({\mathcal E})^\star)$.

\n
Moreover, eq. \eqref{intform1} has a maximal solution $\Psi^{\rm{max}}$
defined on an interval of the form $[0, T^\star)$, and the following ``blow-up
alternative''
holds: either $T^\star = \infty$ or
$$
\lim_{t \to T^\star} \| \Psi_t^{\rm{max}} \|_{{\mathcal D} ({\mathcal
    E})}
\ = \ + \infty,
$$
where we denoted by $\Psi_t^{\rm{max}}$ the function $\Psi^{\rm{max}}$ evaluated at time $t$.
\end{prop}
\begin{proof}
We define the space
$
{\mathcal X} : = L^\infty ([0,T), {\mathcal D} ({\mathcal E})),$
endowed with the norm
$
\| \Psi \|_{\mathcal X} \ : = \ \sup_{t \in [0,T)} \| \Psi_t
\|_{{\mathcal D} ({\mathcal E})}.
$
Given $\Psi_0 \in {\mathcal D} ({\mathcal E})$, we define the map $G : {\mathcal X}
\longrightarrow {\mathcal X}$ as
$$
G \Phi : = e^{- i H \cdot} \Psi_0 + i \int_0^\cdot
e^{- i H (\cdot - s)} | \Phi_s |^2 \Phi_s \, ds.
$$
Notice that the nonlinearity preserves the space ${\mathcal D} ({\mathcal E})$. Indeed,
since any component $\psi_i$ of $\Psi$
belongs to $H^1 (\erre^+)$, then $
|\psi_i|^2 \psi_i$ belongs to $L^2 (\erre^+)$ too, and so the energy
space
for the delta-prime case is preserved. Furthermore, the
product preserves the continuity at zero required by
the
Kirchhoff and the delta case.

By estimates \eqref{gajardo} one
obtains
$$
\| | \Phi_s |^2 \Phi_s \|_{{\mathcal D} ({\mathcal E})} \ \leqslant \ C \| \Phi_s \|_{{\mathcal D} ({\mathcal E})}^3,
$$
so
\begin{equation}
\label{contraz1}
\begin{split}
\| G \Phi \|_{\mathcal X} \ \leqslant \ & \| \Psi_0 \|_{{\mathcal D}
  ({\mathcal E})} + C \int_0^T
\| \Phi_s \|_{{\mathcal D} ({\mathcal E})}^3 \, ds \
\leqslant \ \| \Psi_0 \|_{{\mathcal D} ({\mathcal E})} + C T \| \Phi \|_{\mathcal X}^3\,.
\end{split}
\end{equation}
Analogously, given $\Phi, \Xi \in {\mathcal D} ({\mathcal E})$,
\begin{equation}
\label{contraz2}
\begin{split}
\| G \Phi - G \Xi \|_{\mathcal X} \ \leqslant \ & C T
\left( \| \Phi \|_{\mathcal X}^2 + \| \Xi \|_{\mathcal X}^2 \right)
\| \Phi - \Xi \|_{\mathcal X}\,.
\end{split}
\end{equation}
We point out that the constant $C$ appearing in \eqref{contraz1} and
\eqref{contraz2} is independent of $\Psi_0$, $\Phi$, and $\Xi$.
Now let us restrict the map $G$ to elements $\Phi$ such that $\| \Phi
\|_{\mathcal X} \leqslant 2 \| \Psi_0 \|_{{\mathcal D} ({\mathcal
    E})}$. From \eqref{contraz1} and \eqref{contraz2}, if
$T$ is chosen to be strictly less than $(8C \| \Psi_0 \|_{{\mathcal D}
({\mathcal E})}^2)^{-1}$, then
 $G$ is a contraction of the ball in ${\mathcal X}$ of radius
$ 2 \| \Psi_0 \|_{{\mathcal D} ({\mathcal E})}$, and so,
by the contraction lemma,
there exists a unique solution to \eqref{intform1} in the
time
interval $[0, T)$. By a standard one-step boostrap argument one
immediately has that the solution actually belongs to $C^0 ([0,T),
{\mathcal D} ({\mathcal E}))$, and
due to the validity of \eqref{diffform} in the space
${\mathcal D} ({\mathcal E})^\star$ we immediately have that the solution
$\Psi$ actually belongs to $C^0 ([0,T), {\mathcal D} ({\mathcal E}))) \cap
C^1 ([0,T),{\mathcal D} ({\mathcal E})^\star)$.

The proof of the
existence of a maximal solution is standard, while
the blow-up alternative is a consequence of the fact that,
whenever the ${\mathcal D} ({\mathcal E})$-norm of the solution is
finite, it is possible to extend it for a further time by the same
contraction
argument.
\end{proof}

The next step consists in the proof of the conservation laws.
\begin{prop}
For any solution $\Psi \in C^0 ([0,T), {\mathcal D} ({\mathcal E}))
\cap C^1 ([0,T), {\mathcal D} ({\mathcal E})^\star)$ to
the problem \eqref{intform1}, the following conservation laws hold at
any time $t$:
\begin{equation*}
\| \Psi_t \| \ = \ \| \Psi_0 \|, \qquad
{\mathcal E} ( \Psi_t ) \ = \ {\mathcal E} ( \Psi_0 ),
\end{equation*}
where the symbol $ {\mathcal E}$ denotes the {\em energy functional}
\begin{equation*}
{\mathcal E} ( \Psi_t ) \ : = \ \f 1 2 {\mathcal E}_{lin} ( \Psi_t )
-
\f 1 4 \| \Psi_t \|_{L^4}^4.
\end{equation*}
Here the functional ${\mathcal E}_{lin}$ coincides with ${\mathcal
  E}_{F}$, ${\mathcal E}_{\delta}^\alpha$ or ${\mathcal
  E}_{\delta^\prime}^\beta$, according to the case one considers.
\end{prop}

\begin{proof}
The conservation of the $L^2$-norm can be immediately obtained by the
validity of equation \eqref{diffform} in the space ${\mc D}({\mc E})$:
$$
\f d {dt} \| \Psi_t \|^2 \ = \ 2 \, {\rm{Re}} \, \left\langle \Psi_t ,
\f d {dt} \Psi_t \right\rangle \ = \ 2 \, {\rm{Im}} \, \langle \Psi_t ,
H \Psi_t \rangle \ = \ 0
$$
by the self-adjointness of $H$. In order to prove the conservation of the
energy, first we notice that
$\langle \Psi_t, H \Psi_t \rangle$ is differentiable as a function of time.
Indeed,
\begin{equation*}
\begin{split} &
\f 1 h \left[ \langle \Psi_{t+h}, H \Psi_{t+h} \rangle -
\langle \Psi_{t}, H \Psi_{t} \rangle \right] 
\ = \ \left\langle \f{\Psi_{t+h} - \Psi_t} h, H \Psi_{t+h}
\right\rangle + \left\langle H \Psi_{t} ,\f{\Psi_{t+h} - \Psi_t} h
\right\rangle
\end{split}
\end{equation*}
and then, passing to the limit $h\to0$,
\begin{equation}
\label{previous}
\f d {dt} (\Psi_t ,
H \Psi_t)\ = \ 2 \, {\rm{Re}} \, \left\langle \f d {dt} \Psi_t ,
H \Psi_t \right\rangle \ = \ 2 \, {\rm{Im}} \, \langle | \Psi_t |^2
\Psi_t , H \Psi_t \rangle,
\end{equation}
where we used the self-adjointness
of $H$ and \eqref{diffform}.
Furthermore,
\begin{equation}
\label{naechst}
\f d {dt} (\Psi_t ,
| \Psi_t |^2 \Psi_t) \ = \ \f d {dt} (\Psi_t^2 , \Psi_t^2)
 \ = \ 4 \, {\rm{Im}} \,
\langle | \Psi_t |^2
\Psi_t , H \Psi_t \rangle.
\end{equation}
From \eqref{previous} and \eqref{naechst} one then obtains
$$
\f d {dt} {\mathcal E} (\Psi_t) \ = \ \f 1 2 \f d {dt} \langle \Psi_t ,
H \Psi_t \rangle - \f 1 4 \f d {dt} (\Psi_t ,
| \Psi_t |^2 \Psi_t)_{L^2} \ = \ 0
$$
and the proposition is proved.
\end{proof}
\begin{cor}
The solutions are globally defined in time.
\end{cor}

\begin{proof}
By estimate \eqref{gajardo} with $p = \infty$ and conservation of
the $L^2$-norm, there exists
a constant $M$, that depends on $\Psi_0$ only, such that
$$
{\mathcal E} (\Psi_0) \ = \ {\mathcal E} (\Psi_t) \ \geq \
\f 1 2 \| \Psi_t^\prime \|^2 - M \| \Psi_t^\prime \|
$$
Therefore a uniform (in $t$) bound on
$ \| \Psi_t^\prime \|^2$ is obtained. As a consequence,
one has that no blow-up in finite
time can occur, and therefore, by the blow-up alternative, the
solution
is global in time.
\end{proof}

\section{Main Result}
In this section we describe the asymptotic dynamics of a particular
initial state, which resembles
a soliton for the standard cubic NLS on the line.

According to section \ref{sec:wp}, we use the symbol $H$ to
generically
denote the linear part of the evolution, regardless of the fact that
we are considering the Kirchhoff, delta, or delta-prime boundary
conditions. When necessary, we will distinguish between the three of them.

We use the notation
\begin{equation*}
\phi (x) \ = \ \sqrt 2 \cosh^{-1} x\,,\quad x\in\RE\,,
\end{equation*}
and for any $x_0\in\RE$ and $v\in\RE$ we define
\begin{equation}
\label{fix0t}
\phisol{x_0,v}(x,t) = e^{i \f v 2 x} {e^{- i t \f {v^2} 4}}e^{it}
\phi ( x - x_0 - vt)\,,\quad x\in\RE,\ t\in\RE\,.
\end{equation}
The function $\phisol{x_0,v}$ represents a soliton for the cubic NLS on the line
which at time $t=0$ is centered in $x=x_0$ and has velocity $v$.
Therefore,
$\phisol{x_0,v}$ is the solution of the integral equation
\begin{equation}
\label{eq:sol}
\phisol{x_0,v}(x,t)
=
 \big[U_t e^{i \f v 2 \cdot} \phi ( \cdot - x_0)\big](x)
+i \int_0^t \big[U_{t-s}|\phisol{x_0,v}(\cdot,s)|^2\phisol{x_0,v}(\cdot,s)\big](x)ds\,,
\quad t\in\RE
\,.
\end{equation}

Let $\chi$ be the cut off function $\chi \in C^\infty (\erre^+ )$,
$\chi = 1$ in $(2, + \infty)$ and $\chi = 0$ in $(0,1)$. For later use
we define also $\chi_+ = \chi_{[0, +\infty)}$ and $\chi_- = \chi_{(
    -\infty,0]}$, where $\chi_{[a,b]}$ denotes the characteristic
function of the interval $[a,b]$.

\n
Moreover let $x_0$ and $v$ be two positive constants and $0<\de<1$.

\n
We take as initial condition the following function
\begin{equation}
\label{init}
\Psi_0 (x) \ = \ \left( \begin{array}{c} \chi (x) e^{-i
      \f v 2 x} \phi (x - x_0) \\ 0 \\ 0 \end{array} \right)\,,
\qquad
x_0 \ \geqslant \ v^{1 - \delta}\,,
\end{equation}
and we denote by $\Psi_{H,t}$ the solution of the equation
\begin{equation}
\label{intform}
\Psi_{H,t} \ = \ e^{-iH t} \Psi_0 + i \int_0^t ds \, e^{-i
  H (t-s)} | \Psi_{H,s} |^2 \Psi_{H,s}\,,\quad t\geqslant 0\,.
\end{equation}
The choice of the vector $\Psi_0$ is used to render the idea that the
initial condition is a soliton centered away from the vertex
and moving towards the vertex with velocity $v$. The cut off function
$\chi$ in formula \eqref{init} is aimed at setting $\Psi_0$ in the domain
of the Hamiltonian $H$, see section \ref{subsec2.1}.

Let us set $t_2:=x_0/v+v^{-\delta}$ and define the following functions:
\begin{equation}
\label{Phi1t2}
\Phi^{1}_{H,t_2}(x_1 ,x_2, x_3)\equiv
\begin{pmatrix}
\tilde r_H e^{-i \f {v^2} 4 t_2}e^{i\f v 2 x_1} e^{it_2} \phi (x_1 + x_0 - vt_2)\\
\tilde r_H e^{-i \f {v^2} 4 t_2}e^{-i\f v 2 x_2} e^{it_2} \phi (x_2 - x_0 + vt_2)\\
0
\end{pmatrix}
\end{equation}
\begin{equation}
\label{Phi2t2}
\Phi^{2}_{H,t_2}(x_1 ,x_2, x_3)\equiv
\begin{pmatrix}
0\\
\tilde t_H
e^{-i \f {v^2} 4 t_2}e^{i\f v 2 x_2} e^{it_2} \phi (x_2 + x_0 - vt_2)\\
\tilde t_H e^{-i \f {v^2} 4 t_2}e^{-i\f v 2 x_3} e^{it_2} \phi (x_3 - x_0 + vt_2)
\end{pmatrix}
\end{equation}
\begin{equation}
\label{Phi3t2}
\Phi^{3}_{H,t_2}(x_1 ,x_2, x_3)\equiv
\begin{pmatrix}
\tilde t_H e^{-i \f {v^2} 4 t_2}e^{-i\f v 2 x_1} e^{it_2} \phi (x_1 - x_0 + vt_2)\\
0\\
\tilde t_H e^{-i \f {v^2} 4 t_2}e^{i\f v 2 x_3} e^{it_2} \phi (x_3 + x_0 - vt_2)
\end{pmatrix}
\end{equation}
They represent solitons on the line multiplied
by the scattering coefficients of the linear dynamics $\tilde r_H$ and
$\tilde t_H$, that, in the particular regime we consider, are defined as
follows:

\begin{equation}
\label{tildscoeff}
\begin{split}
& \tilde r_{H_F} \ = \ - 1/3, \qquad \tilde t_{H_F} \ = \ 2/3 \\
& \tilde r_{H_\delta^{\tilde \alpha v}} \ = \ - \f {1 + 2 i \tilde \alpha}
{3 + 2 i \tilde \alpha}, \qquad
\tilde t_{H_\delta^{\tilde \alpha v}} \ = \ \f 2
{3 + 2 i \tilde \alpha} \\ &
\tilde r_{H_{\delta^\prime}^{\tilde \beta / v}} \ = \ \f{\tilde {\beta} + 2 i}
{\tilde \beta + 6 i},
 \qquad
\tilde t_{H_{\delta^\prime}^{\tilde \beta/v}} \ = \ - \f {4 i}
{\tilde \beta + 6 i } \, .
\end{split}
\end{equation}

\begin{rem}{\em 
Notice that the coefficients $\tilde r_H$ and $\tilde t_H$ can be obtained by the
scattering coefficients \eqref{scoeffd},
\eqref{scoeffdp}, identifying $k$ with $v/2$ and
replacing $\alpha$ by $ \tilde \alpha v$
and $\beta$ by
$\tilde \beta / v$. This is due to the fact that we implicitly considered a particle
with mass equal to $1/2$, therefore the momentum $k$ is linked to the
speed $v$ by $k = v/2$. 
}
\end{rem}
For any $t>t_2$ we define the vectors $\Phi^j_{H,t}$ as the evolution of $\Phi^j_{H,t_2}$
with the nonlinear flow generated by $H_j$, i.e., they are solutions of the
equation
\begin{equation}
\label{Phijt}
\Phi^j_{H,t}=e^{-i H_j (t-t_2)} \Phi^j_{H,t_2} + i \int_{t_2}^t ds \, e^{-i
  H_j (t-s)} | \Phi^j_{H,s} |^2 \Phi^j_{H,s}\qquad j=1,2,3\,.
\end{equation}
\begin{rem} {\em
The vectors $\Phi^j_{H,t}$ can be represented by
\begin{equation}
\label{Phi123t}
\Phi^1_{H,t}(x_1 ,x_2, x_3)=
\begin{pmatrix}
e^{-i \f {v^2} 4 t_2}e^{it_2} \phi^{ref}_{t-t_2}(x_1 )\\
e^{-i \f {v^2} 4 t_2} e^{it_2} \phi^{ref}_{t-t_2}(-x_2)\\
0
\end{pmatrix}
\;;\quad
\Phi^2_{H,t}(x_1 ,x_2, x_3)=
\begin{pmatrix}
0\\
e^{-i \f {v^2} 4 t_2} e^{it_2} \phi^{tr}_{t-t_2} (x_2)\\
e^{-i \f {v^2} 4 t_2} e^{it_2} \phi^{tr}_{t-t_2} (-x_3)
\end{pmatrix}
\end{equation}
\begin{equation}
\label{Phi1234t}
\Phi^3_{H,t}(x_1 ,x_2, x_3)=
\begin{pmatrix}
e^{-i \f {v^2} 4 t_2} e^{it_2} \phi^{tr}_{t-t_2} (-x_1)\\
0\\
e^{-i \f {v^2} 4 t_2} e^{it_2} \phi^{tr}_{t-t_2} (x_3)
\end{pmatrix}\,
\end{equation}
where, for any $t\geqslant0$, the functions $\phi^{ref}_t$ and $\phi^{tr}_t$ are the
solutions to the following NLS on the line
\begin{equation}
\label{phiref}
\phi^{ref}_t(x)= \tilde r_H \int_{-\infty}^\infty U_{t}(x-y) e^{i\f v 2 y} \phi (y - v^{1-\de})dy
+i\int_{0}^tds \int_{-\infty}^\infty U_{t-s}(x-y)|\phi^{ref}_s(y)|^2 \phi^{ref}_s(y)dy
\end{equation}
\begin{equation}
\label{phitr}
\phi^{tr}_t(x)= \tilde t_H\int_{-\infty}^\infty U_{t}(x-y) e^{i\f v 2 y} \phi (y - v^{1-\de})dy
+i\int_{0}^tds \int_{-\infty}^\infty U_{t-s}(x-y)|\phi^{tr}_s(y)|^2 \phi^{tr}_s(y)dy\,.
\end{equation}}
\end{rem}

Our main result is summarized in the following theorem:
\begin{theorem}
\label{mainth}

Fixed $\tilde \alpha, \tilde \beta > 0$,
let $H$ be any of the self-adjoint operators $H_F$, $H_\delta^{\tilde \alpha
  v}$, $H_{\delta^\prime}^{\tilde \beta/v}$ acting on $L^2(\mc G)$, where ${\mc G}$ is the three-edge star graph,
and defined by
\eqref{san}, \eqref{fuji}, \eqref{domdelta},
\eqref{domprime}.
Call $\Psi_t$ the unique, global solution to the Cauchy problem
\eqref{intform} with initial
data \eqref{init}.

Then,
there exist $\tau_* >0$ and $T_{\ast}>0 $ such that for $0< T <T_{\ast}$ one has
\begin{equation} \label{ergebnis}
\| \Psi_t - \Phi^1_t- \Phi^2_t- \Phi^3_t \|_{L_x^2(\GG)} \ \leqslant \
C v^{-\f{T_{\ast}-T}{\tau_*}}
\end{equation}
for any time $t$ in the interval $x_0 / v + v^{1 - \delta} < t <
x_0 / v + v^{1 - \delta}
+ T\ln v$.

In \eqref{ergebnis}, $C$ is a positive constant independent of $t$ and $v$,
the
functions
$\Phi_t^j$ are defined by formulas \eqref{fittilde},
$\tilde t$ and $\tilde r$, given in \eqref{tildscoeff}, are the scattering coefficients associated to $H$.
\end{theorem}
The proof of the theorem will be broken into three steps or equivalently
we break the time evolution of $\Psi_{H,t}$ into three phases.

\begin{rem}{\em
A further consequence of Theorem \ref{mainth}, as in the case of \cite{[HMZ07]},
is the fact that fast solitons have reflection and transmission
coefficients which, up to negligible corrections, coincide with the
corresponding coefficients of the linear graph. For example, a definition of
the transmission coefficients along the branches $j=2,3$ could be
given considering the ratio between the amount of mass on the $j-$edge
and the total mass, in the limit $t\to \infty$:
$$
\lim_{t\to \infty}\frac{\|\Psi^j_t\|}{\|\Psi_t\|}\equiv |t^j(v)|\ ,
$$
where $\Psi^j$ denotes the restriction of the solution to the $j-$edge.
In our case we do not have at our disposal the rigorous asymptotics for $t\to \infty$; nevertheless, we can obtain a weaker result. We have the results of Theorem 1.1, which give, in the time interval $t_2 < t < t_2 + T\ln v\ ,$ the estimate
$$
\frac{\|\Psi^j_t\|}{\|\Psi_t\|}=|\tilde t^j| + {\mathcal{O}}(v^{-\sigma})
$$
for a certain $\sigma>0$ and where $\tilde t^j$ is the scattering coefficient of the {\it linear} Hamiltonian which describes the vertex. So, in the limit of fast solitons, i.e. $v\to \infty$, one can assert that the ratio which defines the nonlinear scattering coefficient converges to the corresponding linear scattering coefficient. And analogously
for the case of reflection coefficient $r(v)$, i.e. $j=1$, one has (for $t>t_2$, as before)
$$
\lim_{v\to \infty}\frac{\|\Psi^1_t\|}{\|\Psi_t\|}=|\tilde r|\ .
$$
\p
This is true for every coupling between the ones considered, i.e. Kirchhoff, $\delta$ or $\delta'$.}\par
\end{rem}

\section{Proof of theorem \ref{mainth}}
In the proof we drop the subscript $H$. When convenient, we
specify the particular Hamiltonian operator we refer to.

For any $t\in\RE$ we introduce the soliton
\begin{equation}
\label{truesol}
\Phi_t (x_1 ,x_2, x_3) \ = \ \left( \begin{array}{c} e^{-i \f {v^2} 4 t}
e^{-i
      \f v 2 x_1} e^{it} \phi (x_1 - x_0 + vt) \\
       e^{-i \f {v^2} 4 t}
e^{i
      \f v 2 x_2} e^{it} \phi (x_2 + x_0 - vt)
 \\ 0 \end{array} \right)\,.
\end{equation}
Then, the function $\Phi_t$ satisfies the equation
\begin{equation}
\label{duhamsol}
\Phi_t \ = \ e^{-i H_1 t} \Phi_0 + i \int_0^t ds \, e^{-i
  H_1 (t-s)} | \Phi_s|^2 \Phi_s\,,\quad t\geqslant 0\,.
\end{equation}

\subsection{Phase 1}
We call ``phase 1'' the dynamics in the
time interval $\left(0, t_1 \right)$ with $t_1= \f {x_0} v - v^{-\delta}$.
In this interval we approximate the solution by
the soliton \eqref{truesol}.
The content of this subsection is the estimate of the error due to
such an approximation, that is contained in proposition \ref{fase1}.
Before proving it, we need two
lemmas.

\begin{lemma}
\label{caudale}
Given $0 \leqslant t_a \leqslant t_b \leqslant t_1$,
for the functions
\begin{equation}
\begin{split}
\label{cappa}
K_1 (t,x)& \ : = \ \int_0^\infty U_{t-t_a} (x + y) e^{-i \f v 2 y} \phi ( y - x_0)
\, dy + i \int_{t_a}^t ds
\int_0^\infty U_{t-s} (x + y) e^{-i \f v 2 y}
{e^{- i s \f {v^2} 4}} e^{is} \phi^3 ( y - x_0 + vs)
\, dy \\
K_2 (t,x) & \ : = \ \int_0^\infty U_{t-t_a} (x + y) e^{i \f v 2 y} \phi ( y + x_0)
\, dy + i \int_{t_a}^t ds
\int_0^\infty U_{t-s} (x + y) e^{i \f v 2 y}
{e^{- i s \f {v^2} 4}} e^{is} \phi^3 ( y + x_0 - vs)
\, dy
\end{split}
\end{equation}
the following estimate holds:
\begin{equation*}
\| K_i \|_{X_{t_a, t_b} (\erre^+)}
\leqslant \ C e^{-x_0 + v t_b}, \qquad i = 1,2,
\end{equation*}
where $X_{t_a, t_b}(\erre^+) : = L^\infty_{[t_a, t_b]} L^2 (\erre^+)
\cap L^6_{[t_a, t_b]} L^6 (\erre^+)$.
\end{lemma}
\begin{proof}
Let us start with $K_1$. Adding and subtracting a
contribution to negative values of $y$ one can write
\begin{equation}
 \begin{split}
K_1 (t,x) & \
= \ \int_{-\infty}^\infty U_{t-t_a} (x + y) e^{-i \f v 2 y} \phi ( y - x_0)
\, dy + i \int_{t_a}^t ds
\int_{-\infty}^\infty U_{t-s} (x + y) e^{-i \f v 2 y}
{e^{- i s \f {v^2} 4}} e^{is} \phi^3 ( y - x_0 + vs)
\, dy \\
& \
- \int_{-\infty}^0 U_{t-t_a} (x + y) e^{-i \f v 2 y} \phi ( y - x_0)
\, dy - i \int_{t_a}^t ds
\int_{-\infty}^0 U_{t-s} (x + y) e^{-i \f v 2 y} {e^{- i s \f {v^2} 4}}
e^{is} \phi^3 ( y - x_0 + vs)
\, dy \\
= & \ e^{i \f v 2 x} {e^{- i t \f {v^2} 4}}e^{it}
\phi ( x + x_0 - vt) \\
& \
- \int_{-\infty}^0 U_{t-t_a} (x + y) e^{-i \f v 2 y} \phi ( y - x_0)
\, dy - i \int_{t_a}^t ds
\int_{-\infty}^0 U_{t-s} (x + y) e^{-i \f v 2 y} {e^{- i s \f {v^2} 4}}
e^{is} \phi^3 ( y - x_0 + vs)
\, dy
\end{split}
\label{terms}
\end{equation}
where we used the integral equation \eqref{eq:sol}.
By a straightforward computation,
the $X_{t_a, t_b}(\erre^+)$-norm of the first term can be bounded by
$
C e^{ -(x_0 - vt_b) }.
$
To evaluate the size of the second term, let us write it as follows:
\begin{equation}
\label{kaka}
 \int_0^\infty U_{t-t_a} (x-y) e^{i \f v 2 y} \phi (y + x_0)dy
\ =[ U^-_{t-t_a} e^{i \f v 2 \cdot} \phi (\cdot + x_0)] (x)=
\ [U_{t-t_a} \chi_+ e^{i \f v 2 \cdot} \phi (\cdot + x_0)] (x), \quad x > 0.
\end{equation}
Using the one-dimensional homogeneous Strichartz's estimates for $U_{t-t_a}$,
namely, the analogous of \eqref{stric1} for functions of the half
line, we can estimate the $X_{t_a, t_b}(\RE^+)$-norm of this term as
$$
\|U_{t-t_a} \chi_+ e^{i \f v 2 \cdot} \phi (\cdot + x_0) \|_{X_{t_a, t_b}(
  \RE )} \leqslant
C \| \chi_+ e^{i \f v 2 \cdot} \phi (\cdot + x_0) \| \ \leqslant C
e^{-x_0},
$$
where we used the notation
$X_{t_a, t_b} (\erre) : = L^\infty_{[t_a, t_b]} L^2 (\erre) \cap
L^6_{[t_a, t_b]} L^6 (\erre)$.

\n
The norm of the last term on the r.h.s. of equation \eqref{terms} can be estimated in a similar way by
\begin{equation}
\label{akak}
\left\| \int_{t_a}^{\cdot} \big[ U_{\cdot-s} |\phi_{-x_0 , v} (s)|^2
\phi_{-x_0 , v} (s) \big]ds \right\|_{X_{t_a, t_b}(\RE^+)}
 \ \leqslant C \| \phi_{-x_0 , v}^3\|_{L^1_{[t_a,t_b]} L^2(\RE^+)} \leqslant C\f {1} {v} e^{-x_0 + vt_b}.
\end{equation}
Therefore, from \eqref{terms}, \eqref{kaka}, and \eqref{akak} we get
\begin{equation*}
\| K_1 \|_{X_{t_a, t_b}(\erre^+)} \
\leqslant \ C e^{-x_0 + vt_b}.
\end{equation*}

To estimate $K_2$, the first term
in its definition \eqref{cappa}
can be treated as in \eqref{kaka}, while the second is
estimated following the line of \eqref{akak}.
\end{proof}

\begin{lemma}
\label{blocco}
Given $0 \ls t_a \ls t_b \ls t_1$, let $a$ and $b$ two strictly positive numbers, with
\begin{equation*}
b \leqslant \f 1 {8 a^2 + 4 a}.
\end{equation*}
Moreover, let $y$ be a real, continuous function such that
$0 \leqslant y (t_a) \leqslant a$, and
\begin{equation}
\label{constr}
0 \ls
y (t) \ls a + b y^2 (t) + b y^3 (t), \qquad {\mbox{for any }} \ t \in [t_a, t_b].
\end{equation}
Then,
\begin{equation*}
\max_{t \in [t_a, t_b]} y (t) \, \ls \, 2 a.
\end{equation*}
\end{lemma}

\begin{proof}
Consider the function $f_b(x) : = b x^3 + bx^2 - x + a$.
 Denoted $\bar b := \f 1 {8
  a^2 + 4 a}$, one has 
$
f_{\bar b} (2a) \ = 0
$.

\n
If $b \ls \bar b$, then $f_b (x) \ls f_{\bar b} (x)$ for any $x >
0$. Besides, notice that $f_b (0) = a > 0$, then there must be a point $\tilde x
\in (0, 2a]$ s.t. $f_b (\tilde x) = 0$. Finally, since the
function $y$ is continuous, in order to satisfy the constraint
\eqref{constr} one must have
$$
y (t) \ls \tilde x \ls 2a, \qquad {\mbox{for any }} \ t \in [t_a, t_b].
$$
\end{proof}

\begin{prop}
Let $\Psi_t$ be the solution of the equation \eqref{intform}, and
$\Phi_t$ be the solution of equation \eqref{duhamsol}.
There exists $C > 0$, independent of $t$ and $v$, such that
\label{fase1}
\begin{equation}
 \label{error1}
\| \Psi_t - \Phi_t \| \ \leqslant \ C e^{-v^{1 - \delta}}
\end{equation}
for any $t \in [0, t_1]$.
\end{prop}
\begin{proof}

\n
Let us define $\Xi_t : = \Psi_t - \Phi_t$, and fix $t_a \in [0,
t_1]$.
Then, from equations \eqref{intform}
and \eqref{duhamsol}, we have
\begin{equation*}
\begin{split}
\Xi_t \ = & \ e^{-iH (t-t_a)} \Xi_{t_a} + (e^{-iH (t-t_a)} - e^{-i H_1 (t-t_a)})
\Phi_{t_a} + i \int_{t_a}^t (e^{-iH (t-s)} - e^{-i H_1 (t-s)}) |
\Phi_s |^2 \Phi_s \, ds\\
& \ + i \int_{t_a}^t e^{-iH (t-s)} \left[ | \Xi_s |^2 \Xi_s + | \Xi_s
  |^2 \Phi_s + | \Phi_s |^2 \Xi_s + 2 {\mbox{Re}} ( \overline{\Xi_s} \Phi_s) \Xi_s
+ 2 {\mbox{Re}} ( \overline{\Xi_s} \Phi_s) \Phi_s \right]\\
= & \ e^{-iH (t-t_a)} \Xi_{t_a} + F(t_a,t)\\
& + i \int_{t_a}^t e^{-iH (t-s)} \left[ | \Xi_s |^2 \Xi_s + | \Xi_s
  |^2 \Phi_s + | \Phi_s |^2 \Xi_s + 2 {\mbox{Re}} ( \overline{\Xi_s} \Phi_s) \Xi_s
+ 2 {\mbox{Re}} ( \overline{\Xi_s} \Phi_s) \Phi_s \right],
\end{split}
\end{equation*}
where we defined
\begin{equation*}
F (t_a, t) \ : = \ (e^{-iH (t - t_a)} - e^{-i H_1 (t - t_a)})
\Phi_{t_a} + i \int_{t_a}^{t} (e^{-iH (t-s)} - e^{-i H_1
  (t -s)}) |
\Phi_s |^2 \Phi_s \, ds\,.
\end{equation*}
Let us fix $t_b \in [t_a,t_1]$, and denote
$X_{t_a, t_b} = L^{\infty}_{[t_a, t_b] } L^2
\cap L^{6}_{[t_a, t_b] } L^6$.
Then
\begin{equation}
\begin{split}
\label{split}
\| \Xi \|_{X_{t_a, t_b}} \ \ls &
 \ \| e^{-iH (\cdot -t_a)} \Xi_{t_a} \|_{X_{t_a, t_b}} + \| F(t_a, \cdot) \|_{X_{t_a, t_b}} \\
 &+ \left\| \int_{t_a}^\cdot e^{-iH (\cdot-s)} \left[ | \Xi_s |^2 \Xi_s + | \Xi_s
  |^2 \Phi_s + | \Phi_s |^2 \Xi_s + 2 {\mbox{Re}} ( \overline{\Xi_s} \Phi_s) \Xi_s
+ 2 {\mbox{Re}} ( \overline{\Xi_s} \Phi_s) \Phi_s \right] \right\|_{X_{t_a, t_b}} .
\end{split}
\end{equation}
Using \eqref{stric1} the first term on the r.h.s can be estimated as
$$
 \| e^{-iH (\cdot -t_a)} \Xi_{t_a} \|_{X_{t_a, t_b}} \ \ls \ C \|
 \Xi_{t_a} \|.
$$

\n
We estimate the integral term on the r.h.s. of \eqref{split} also
using Strichartz's estimates. Let us analyse in detail the cubic term.
Since both pairs of indices $(\infty ,2 )$ and $(6,6)$ fulfil \eqref{admissible},
in \eqref{stric2} we can choose $q=k= 6/5$ and obtain
\begin{equation}
\lf\|
\int_{t_a}^{\cdot} e^{-iH (\cdot-s)} | \Xi_s |^2 \Xi_s \,ds
\ri\|_{X_{t_a, t_b}}
\leqslant C \| | \Xi_{\cdot} |^2 \Xi_{\cdot} \|_{ \LT{6/5}{t_a,t_b} \LG{6/5} }.
\label{yume}
\end{equation}
Moreover, by standard H\"older estimates,
\begin{equation}
\label{hoel1}
\| | \Xi_{\cdot} |^2 \Xi_{\cdot} \|_{ \LT{6/5}{t_a,t_b} \LG{6/5} }\ \leqslant
\ \left\| \| \Xi_\cdot \|^2_{L^6} \| \Xi_\cdot \|_{L^2} \right\|_{\LT{6/5}{t_a,t_b}}
\ \ls \ \| \Xi \|^2_{ \LT{6}{t_a,t_b} \LG{6} } \| \Xi \|_{ \LT{2}{t_a,t_b} \LG{2} }
\ \ls \ ( t_b - t_a)^{\f 1 2} \| \Xi \|^2_{ \LT{6}{t_a,t_b} \LG{6} }
\| \Xi \|_{ \LT{\infty}{t_a,t_b} \LG{2} }.
\end{equation}
Then,
by \eqref{yume} and \eqref{hoel1},
\begin{equation}
\label{eq-dopo-hoel1}
\lf\|
\int_{t_a}^{\cdot} e^{-iH (\cdot-s)} | \Xi_s |^2 \Xi_s \,ds
\ri\|_{X_{t_a, t_b}}
\leqslant C
(t_b-t_a )^{1/2} \|\Xi\|^3_{X_{t_a, t_b}}.
\end{equation}
Notice that the constant $C$ can be chosen independently of $t_a,
t_b$, and of the boundary condition at the vertex.

\n
The other terms in the integral on the r.h.s. of \eqref{split} can be
estimated analogously. One ends up with
\begin{equation*}
\begin{split}
\| \Xi \|_{X_{t_a, t_b}} \ls & \ C
 \ \| \Xi_{t_a} \| + \| F(t_a, \cdot) \|_{X_{t_a, t_b}} + C (t_b -
 t_a)^{\f 1 2} \| \Xi \|^3_{X_{t_a, t_b}} + C (t_b -
 t_a)^{\f 2 3} \| \Xi \|^2_{X_{t_a, t_b}} + C (t_b -
 t_a)^{\f 5 6} \| \Xi \|_{X_{t_a, t_b}},
\end{split}
\end{equation*}
where the arising norms of $\Phi$ were absorbed in the constant $C$.

\n
If $t_b$ and $t_a$ are sufficiently close, then $C (t_b -
 t_a)^{\f 5 6} < 1/2$. Furthermore, since the quantity $t_b - t_a$ is
 upper bounded, one can estimate $(t_b -
 t_a)^{\f 2 3}$ by $C (t_b -
 t_a)^{\f 1 2}$. So, for some $ \widetilde C>0$
\begin{equation}
\label{split3}
\begin{split}
\| \Xi \|_{X_{t_a, t_b}} \ls & \ \widetilde C
 \ \| \Xi_{t_a} \| + \widetilde C \ \| F(t_a, \cdot) \|_{X_{t_a, t_b}} + \widetilde C (t_b -
 t_a)^{\f 1 2} \left[ \| \Xi \|^3_{X_{t_a, t_b}} + \| \Xi \|^2_{X_{t_a, t_b}} \right].
\end{split}
\end{equation}
Applying lemma \ref{blocco} to the function $y(t) = \| \Xi \|_{X_{t_a, t}}$, which is
continuous and monotone, one has that, if
$$
t_b - t_a \ \ls \ \left( 8 \widetilde C^3 \ (
 \ \| \Xi_{t_a} \| + \| F(t_a, \cdot) \|_{X_{t_a, t_b}})^2 + 4
  \widetilde C^2
 \ ( \| \Xi_{t_a} \| + \| F(t_a, \cdot) \|_{X_{t_a, t_b}} ) \right)^{-2},
$$
then $\| \Xi \|_{X_{t_a, t_b}} \ls \ 2 \widetilde C
 \ \| \Xi_{t_a} \| + 2 \wt C \ \| F(t_a, \cdot) \|_{X_{t_a,
     t_b}}$. From the immediate estimates
\begin{equation*}
\| \Xi_{t_a} \| \ \ls \ 4, \qquad
\| F(t_a, \cdot) \|_{X_{t_a, t_b}} \ \ls \ \| F(0, \cdot) \|_{X_{0, t_1}},
\end{equation*}
if one denotes
\begin{equation*}
\tau \ : = \ \left( 8 \widetilde C^3 \ (
 4 + \| F(t_a, \cdot) \|_{X_{0, t_1}})^2 + 4
  \widetilde C^2
 \ ( 4 + \| F(t_a, \cdot) \|_{X_{0, t_1}} ) \right)^{-2},
\end{equation*}
then for any $t \in [0, t_1)$ 
$\| \Xi \|_{X_{t, t + \tau}} \ls \ 2 \widetilde C
 \ \| \Xi_{t} \| + 2 \wt C \ \| F(t, \cdot) \|_{X_{t, t+ \tau}}$.

We divide the interval $[0, t_1]$ in $N + 1$ subintervals as follows
$$
[0, t_1] \ = \ \left( \cup_{j=0}^{N-1} [j \tau,
(j+1) \tau] \right) \cup [N \tau, t_1],
$$
where
$$
N : = \left[ \f {x_0 - v^{1 - \delta}}{v \tau} \right], \qquad
[ \cdot ] = {\mbox{ integer part.}}
$$
Making use of lemma \ref{blocco}, and noting that $\| \Xi_{(j+1) \tau}
\| \ls \| \Xi \|_{X_{j\tau, (j+1) \tau}}$,
one proves by induction
that
\begin{equation}
\label{indu}
\begin{split}
\| \Xi \|_{X_{j\tau, (j+1) \tau}} \ & \ls \ (2 \widetilde C)^{j+1} \| \Xi_0
\| + \sum_{k = 0}^{j} (2 \widetilde C)^{j+1-k} \| F (k\tau, \cdot)
\|_{X_{k\tau, (k+1) \tau}},
\qquad j = 0, \dots , N - 1,
\\
\| \Xi \|_{X_{N\tau, t_1}} \ & \ls \ (2 \widetilde C)^{N+1} \| \Xi_0
\| + \sum_{k = 0}^{N-1} (2 \widetilde C)^{N+1-k} \| F (k\tau, \cdot) \|_{X_{k\tau, (k+1) \tau}}
+ 2 \wt C \ \| F (N\tau, \cdot) \|_{X_{N\tau, t_1 }}\,,
\end{split}
\end{equation}
where the last inequality comes from the fact that $t_1 - N \tau \ls
\tau$, so lemma \ref{blocco} applies to this last step too.

\n
The norm of $\Xi$ as a function of the whole time interval $[0,
t_1]$ can be estimated by
\begin{equation}
\label{almost}
\begin{split}
\| \Xi \|_{X_{0,t_1}} \ \ls & \ \sum_{j=0}^{N-1}
\| \Xi \|_{X_{j\tau, (j+1) \tau}} + \| \Xi \|_{X_{N\tau, t_1}} \\ \ls & \
\sum_{j=0}^{N} (2 \widetilde C)^{j+1} \| \Xi_0
\| + \sum_{j=0}^{N} \sum_{k = 0}^{j} (2 \widetilde C)^{j+1-k} \|
F (k\tau, \cdot) \|_{X_{k\tau, \min\{t_1, (k+1) \tau\}}} \,.
\end{split}
\end{equation}

In order to prove the theorem using \eqref{almost}, we need more precise
estimates for $\| \Xi_0 \|$ and $\| F (j\tau, \cdot) \|_{X_{j\tau, (j+1) \tau}}$.

\n
First,
\begin{equation}
 \| \Xi_0 \|^2 \leqslant \ \int_0^{2} \phi^2
(x-x_0) \, dx + \int_0^\infty \phi^2 (x + x_0)
= \ 2 (1 - {\rm{tanh}} (x_0 - 2)) \ \leqslant \ C e^{-2 x_0}.
 \label{tamago}
\end{equation}
To estimate $\| F (t_a, \cdot) \|_{X_{t_a, t_b}}$ we
specialize $F$ to the three cases under analysis.
From the explicit propagators \eqref{freepro}, \eqref{propdelta},
\eqref{propdeltaprime}, we get
\begin{equation*}
\begin{split}
F_F(t_a, t)\ = & \ \f 1 3 \,
U_{t-t_a}^+ \left( \begin{array} {ccc}
-1 & -1 & 2 \\ -1 & -1 & 2 \\
2 & 2 & 2 \end{array} \right) \Phi_{t_a} + \f i 3 \int_{t_a}^t
U_{t-s}^+ \left( \begin{array} {ccc}
-1 & -1 & 2 \\ -1 & -1 & 2 \\
2 & 2 & 2 \end{array} \right) | \Phi_s |^2 \Phi_s \, ds\, \\
F_\delta^{\tilde \alpha v} (t_a,t)\ = & \ F_F(t_a,t) - \f 2 9 \tilde
\alpha v
\int_0^{+\infty}du 
\, e^{-\f {\tilde \alpha } 3 u v} \left(U_{t-t_a}^+ \, {\mathbb J} \,
\Phi_{t_a}\right) (\cdot + u) \\ &
- i \f 2 9 \tilde \alpha v \int_{t_a}^t ds \int_0^{+\infty}du
\, e^{-\f {\tilde \alpha} 3 u v} \left( U_{t-s}^+ \, {\mathbb J} \, | \Phi_s |^2
  \Phi_s \right) (\cdot + u)
\\
F_{\delta^\prime}^{\tilde \beta/v} (t_a, t)\ = &
U_{t-t_a}^+ \left( \begin{array} {ccc}
1 & -1 & 0 \\ -1 & 1 & 0 \\
0 & 0 & 2 \end{array} \right) \Phi_{t_a} + i \int_{t_a}^t
U_{t-s}^+ \left( \begin{array} {ccc}
1 & -1 & 0 \\ -1 & 1 & 0 \\
0 & 0 & 2 \end{array} \right) | \Phi_s |^2 \Phi_s \, ds\, \\
& - \f {2v} {\tilde \beta} \int_0^{+\infty}du 
\, e^{-\f 3 {\tilde \beta} v u} \left(U_{t-t_a}^+ \, {\mathbb J} \, \Phi_{t_a}
\right) (\cdot + u)
- i \f {2 v} {\tilde \beta} \int_{t_a}^t ds \int_0^{+\infty}du
\, e^{-\f {3} {\tilde \beta} v u} \left(U_{t-s}^+ \, {\mathbb J} \, | \Phi_s |^2
  \Phi_s \right) (\cdot + u)
\end{split}
\end{equation*}
It is immediately seen that
$$
F_F (t_a, t,x ) \ = \ \f 1 3 \left( \begin{array}{c}
- K_1 (x,t) - K_2 (x,t) \\ - K_1 (x,t) - K_2 (x,t) \\
2 K_1 (x,t) + 2 K_2 (x,t),
\end{array} \right) ,
$$
where $K_1$ and $K_2$ were defined in \eqref{cappa}.
Lemma \ref{caudale} yields
\begin{equation}
\| F_F (t_a, \cdot) \|_{X_{t_a, t_b}} \
\leqslant \ C e^{-x_0 + vt_b}.
\label{aka}
\end{equation}

\n
Furthermore, since
\begin{equation*} \begin{split}
F_\delta^{\tilde \alpha v} (t_a,t,x) \ = \ F_F (t_a,t,x ) - \f 2 9
{\tilde \alpha v} \int_0^{+\infty} du \,
 e^{-\f {\tilde \alpha} 3 u v} \left( \begin{array}{c} K_1 (x+u,t) + K_2 (x+u,t) \\
  K_1 (x+u,t) + K_2 (x+u,t) \\ K_1 (x+u,t) + K_2 (x+u,t)\end{array} \right),
\end{split} \end{equation*}
after the change of variable $u \to u v$ we conclude
\begin{equation}
\| F_\delta^{\tilde \alpha v} (t_a, \cdot) \|_{X_{t_a, t_b}} \
\leqslant \ C e^{-x_0 + vt_b}
\label{aka2} .
\end{equation}

\n
Finally,
\begin{equation*} \begin{split}
F_{\delta^\prime}^{\tilde \beta / v} (t_a,t,x) \ = \ \left( \begin{array}{c} K_1 (x,t) - K_2 (x,t)
\\ - K_1 (x,t) + K_2 (x,t) \\ 0
\end{array} \right) - \f {2v} {\tilde \beta} \int_0^{+\infty} du \,
e^{-\f 3 {\tilde \beta} uv}
\left( \begin{array}{c} K_1 (x+u,t) + K_2 (x+u,t) \\
  K_1 (x+u,t) + K_2 (x+u,t) \\ K_1 (x+u,t) + K_2 (x+u,t)\end{array} \right)
\end{split} \end{equation*}
yields, after the change of variable $u \to uv$,
\begin{equation}
\| F_{\delta^\prime}^{\tilde \beta / v} (t_a, \cdot) \|_{X_{t_a, t_b}} \
\leqslant \ C e^{-x_0 + vt_b}.
\label{aka3}
\end{equation}

Now we go back to estimate \eqref{almost}. Due to \eqref{tamago}, \eqref{aka},
\eqref{aka2}, and \eqref{aka3}, and estimating any geometric sum by
the double of its largest term, which is justified if the rate of the
sum is not less than two, we get
\begin{equation}
\label{ziemlich}
\begin{split}
\| \Xi \|_{X_{0,t_1}} \ \ls & \ 2 C \wt C e^{-x_0}
\sum_{j = 0}^N ( 2 \wt C )^j
+ 2 C \wt C e^{-x_0 + v \tau}
  \sum_{j=0}^{N-1} (2 \widetilde C)^{j} \sum_{k = 0}^{j} \left( \f{e^{v\tau}}
    {2 \widetilde C}\right)^{k} \\ & \
+ C e^{-x_0 + v \tau}
  (2 \widetilde C)^{N+1} \sum_{k = 0}^{N-1} \left( \f{e^{v\tau}}
    {2 \widetilde C}\right)^{k} + 2 C \wt C e^{-x_0 + v t_1} \\
 \ \ls & \ 2 C (2 \wt C)^{N+1} e^{-x_0}
  + 4 C \wt C e^{-x_0 + v \tau}
  \sum_{j=0}^{N-1} e^{j v\tau}
+ 8 C \wt C^2 e^{-x_0 + N v \tau}
    + 2 C \wt C e^{-x_0 + v t_1}
\\
 \ \ls & \ 2 C (2 \wt C)^{N+1} e^{-x_0}
  + 8 C \wt C e^{-x_0 + N v \tau}
  + 8 C \wt C^2 e^{-x_0 + N v \tau}
    + 2 C \wt C e^{-x_0 + v t_1}.
\end{split} \end{equation}
Concerning the first term on the r.h.s. of \eqref{ziemlich}, we have
\begin{equation}
\label{zuerst}
 (2 \wt C)^{N} e^{-x_0} \ \ls \ \left( \f{{2 \wt C}^{\f 1 {v
         \tau}}} e \right)^{x_0 - v^{1 - \delta}} e^{ - v^{1 -
     \delta}} \ \ls \ C e^{ - v^{1 -
     \delta}}.
\end{equation}
From \eqref{ziemlich} and \eqref{zuerst} we get $\| \Xi \|_{X_{0,
    t_1}} \ls C e^{ - v^{1 -
     \delta}}$, so \eqref{error1} follows and the proof is concluded.
\end{proof}

\subsection{Phase 2}

We call ``phase 2'' the evolution of the system in the time interval $( t_1 , t_2 )$ with
$t_2= \f {x_0} v + v^{- \delta}$. Let us define the vector
\begin{equation*}
\Phi^S_t:=\Phi^{S,in}_t+\Phi^{S,out}_t
\end{equation*}
with
\begin{equation*}
\Phi^{S,in}_t:=
\begin{pmatrix}
\phisol{x_0,-v}(t)\\ \\
0\\ \\
0
\end{pmatrix}
\;,\qquad
\Phi^{S,out}_t:=
\begin{pmatrix}
\tilde r\,\phisol{-x_0,v}(t)\\ \\
\tilde t\,\phisol{-x_0,v}(t)\\ \\
\tilde t\, \phisol{-x_0,v}(t)
\end{pmatrix}
\,,
\end{equation*}
where the function $\phi_{x_0,v}$ was defined in equation
\eqref{fix0t} and the reflection and transmission coefficients,
$\tilde r$ and $\tilde t$, must be chosen accordingly to the
Hamiltonian $H$ taken in the equation \eqref{intform}. The explicit
expressions of $\tilde r$ and $\tilde t$ in all the cases
$H=H_F,\,H^{\tilde \al v}_\de,\, H^{\tilde \beta / v}_{\de'}$ can be read in
formula \eqref{tildscoeff}.
\p
\begin{figure}[t!]
\begin{center}
\includegraphics[width=0.9\textwidth]{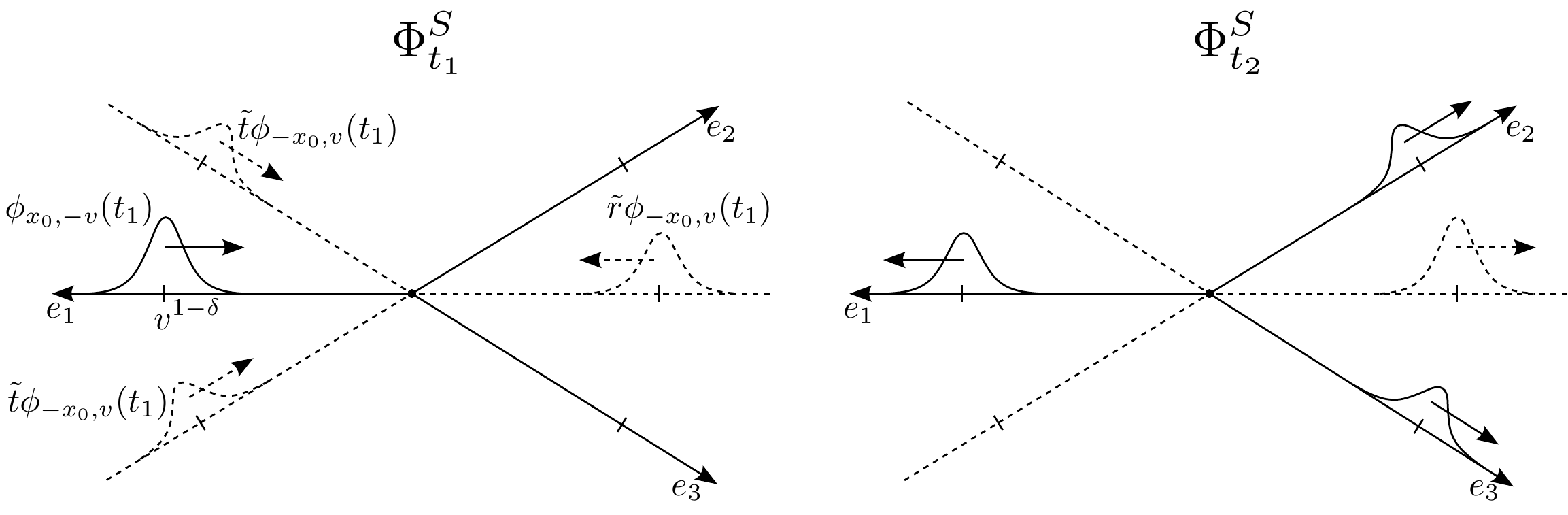}
\caption{\label{fig2}We represent the state $\Phi_t^S$ for $t=t_1$ and
$t=t_2$. The continuous lines stand for the real graph, the dashed lines
represent the extension of each edge to negative values of $x$. At time $t=t_1$
only the vector $\Phi^{S,in}$ is significantly supported on the real graph. At
time $t=t_2$ only the state $\Phi^{S,out}$ has relevant support on the real
graph, while the body of the soliton in $\Phi^{S,in}$ has moved to the negative
axis associated with the first edge. }
\end{center}
\end{figure}

\begin{prop}
\label{lemma2}
Let $t\in \left( t_1,t_2 \right)$ then there exists $v_0>0$ such that for all $v>v_0$
\begin{equation}
\label{prop2.1}
\|\Psi_{t} - \Phi^S_t\|\leqslant C_1 v^{-\f\de2}\,,
\end{equation}
moreover
\begin{equation}
\label{prop2.2}
\|\Psi_{t_2}-\Phi^{S,out}_{t_2}\|\leqslant C_2v^{-\f\de 2}\,,
\end{equation}
where $C_1$ and $C_2$ are positive constants which do not depend on $t$ and $v$.
\end{prop}
\begin{proof}
From the definition of $\Psi_t$, see equation \eqref{intform}, we have
\begin{equation*}
\Psi_{t}=
\ e^{-iH(t-t_1)} \Psi_{t_1} + i \int_{t_1}^{t} ds \, e^{-i
  H(t-s)} | \Psi_s |^2 \Psi_s\,.
\end{equation*}
We start with the trivial estimate
\begin{equation}
\label{ineq1}
\begin{aligned}
&\big\|\Psi_{t}-\Phi^S_t\big\|\\
&\leqslant
\big\|\Psi_{t} - e^{-iH (t-t_1)} \Psi_{t_1}\big\|+\big\|e^{-iH(t-t_1)}\Psi_{t_1}-e^{-iH(t-t_1)}\Phi_{t_1}^{S,in}\big\|+\big\|e^{-iH(t-t_1)}\Phi_{t_1}^{S,in}-\Phi^S_t\big\|
\end{aligned}
\end{equation}
and estimate the r.h.s. term by term. The estimates involved in the analysis of the first term are similar to the ones used in the previous proposition, thus we omit the details. Similarly to what was done above we set $X_{t_1,t_2}= L^{\infty}_{[t_1, t_2] } L^2 \cap L^{6}_{[t_1, t_2] } L^6 $, then by Strichartz estimates (see equation \eqref{eq-dopo-hoel1} and proposition \ref{prop:stric})
\begin{equation}
\label{john}
\big\|\Psi - e^{-iH (\cdot-t_1)} \Psi_{t_1}\big\|_{X_{t_1,t_2}}\leqslant
\bigg\| \int_{t_1}^{\cdot} ds \, e^{-i H (\cdot-s)} | \Psi_s |^2 \Psi_s\bigg\|_{X_{t_1,t_2}}
\leqslant
C(t_2-t_1)^{1/2}\big\|\Psi\big\|_{X_{t_1,t_2}}^3\,,
\end{equation}
\begin{equation*}
\big\|e^{-iH(\cdot-t_1)}\Psi_{t_1}\big\|_{X_{t_1,t_2}}\leqslant C
\big\|\Psi_{t_1}\big\|\,,
\end{equation*}
which imply
\begin{equation*}
\big\|\Psi\big\|_{X_{t_1,t_2}}\leqslant
C\big\|\Psi_{t_1}\big\|+
C(t_2-t_1)^{1/2}\big\|\Psi\big\|_{X_{t_1,t_2}}^3\,.
\end{equation*}
 By lemma \ref{blocco} one has that if $(t_2-t_1)\leqslant\big[8 C^3 \|\Psi_{t_1}\|^2+4C^2 \|\Psi_{t_1}\|\big]^{-2}$, then $\|\Psi\|_{X_{t_1,t_2}}\leqslant 2C\|\Psi_{t_1}\|$; using this estimate in the inequality \eqref{john} we get
\begin{equation}
\label{L2S1}
\big\|\Psi_{t} - e^{-iH (t-t_1)} \Psi_{t_1}\big\|\leqslant
\big\|\Psi - e^{-iH(\cdot-t_1)} \Psi_{t_1}\big\|_{X_{t_1,t_2}}\leqslant
C(t_2-t_1)^{1/2} \big\|\Psi_{t_1}\big\|^3\leqslant Cv^{-\de/2}
\end{equation}
where we used $t_2-t_1=2v^{-\de}$.

We proceed now with the estimate of the second term on the r.h.s. of inequality \eqref{ineq1}. Let us set
\begin{equation*}
\Phi_{t_1}^{tail} :=
\begin{pmatrix}
0\\ \\
\phisol{-x_0,v}(t_1)\\ \\
0
\end{pmatrix}\;.
\end{equation*}
We notice that $\Phi_{t_1}^{S,in}+\Phi_{t_1}^{tail}=\Phi_{t_1}$ where the vector $\Phi_{t}$ was defined in equation \eqref{truesol} and rewrite $\Psi_{t_1}$ by adding and subtracting $\Phi_{t_1}$
\begin{equation*}
\Psi_{t_1}=\Psi_{t_1}-\Phi_{t_1}+\Phi_{t_1}^{S,in}+\Phi_{t_1}^{tail}\,.
\end{equation*}
The following trivial inequality holds true
\begin{equation}
\label{L2S2}
\big\|e^{-iH(t-t_1)}\Psi_{t_1}-e^{-iH(t-t_1)}\Phi_{t_1}^{S,in}\big\|\leqslant
\big\|\Psi_{t_1}-\Phi_{t_1}\big\|+
\big\|\Phi_{t_1}^{tail}\big\|\leqslant Ce^{-v^{1-\de}}
\end{equation}
where in the latter estimate we used proposition \ref{fase1} and the fact that $\|\Phi_{t_1}^{tail}\|\leqslant 2e^{-v^{1-\de}}$.

Let us consider now the last term on the r.h.s. of inequality \eqref{ineq1}. We are going to prove that for all $t\in(t_1,t_2)$ and for $v$ big enough
\begin{equation}
\label{seesee}
\big\|e^{-iH(t-t_1)}\Phi_{t_1}^{S,in}-\Phi^{S}_t\big\|\leqslant Cv^{-\de}\,.
\end{equation}
Let us introduce the functions
\begin{equation}
\label{phi-}
\phi^-_{t}(x):=\int_0^\infty U_{t-t_1}(x-y)\phisol{x_0,-v}(y,t_1) dy
\end{equation}
and
\begin{equation}
\label{phi+}
\phi^+_{t}(x):=\int_0^\infty U_{t-t_1}(x+y)\phisol{x_0,-v}(y,t_1) dy\,.
\end{equation}
First we prove a preliminary formula for the vector $e^{-iH(t-t_1)}\Phi^{S,in}_{t_1}$ (see equations \eqref{carnival} and \eqref{carnival2} below). For any constant $a>0$, not dependent on $v$ let us consider the term
\begin{equation*}
\begin{aligned}
&v a \int_0^{\infty}e^{-uva}\big[U^+_{t-t_1}\phi_{x_0,-v}(t_1)\big](u+x) du\\
=&
v a \int_0^{\infty}du \int_0^\infty dy\, e^{-uv a} U_{t-t_1}(u+x+y)e^{i\varphi(t_1)}e^{-i\frac{v}{2}y}\phi(y-x_0+vt_1)
\end{aligned}
\end{equation*}
where we set $\varphi(t):=- t \f {v^2} 4+t$. By integrating by parts we obtain the equality
\begin{equation}
\begin{aligned}
\label{misunderstood}
&v a\int_0^{\infty}e^{-uv a}\big[U^+_{t-t_1}\phi_{x_0,-v}(t_1)\big](u+x) du\\
=&
2i a\int_0^{\infty}du \int_0^\infty dy\, e^{-uv a} U_{t-t_1}(u+x+y)e^{i\varphi(t_1)}\bigg[\frac{d}{dy}e^{-i\frac{v}{2}y}\bigg]\phi(y-x_0+vt_1)\\
=&A_1(x,t)+A_2(x,t)+A_3(x,t)
\end{aligned}
\end{equation}
with
\begin{equation*}
A_1(x,t) := -
2ia \int_0^{\infty}du\, e^{-uv a} U_{t-t_1}(u+x)e^{i\varphi(t_1)}\phi(-x_0+vt_1)
\end{equation*}
\begin{equation*}
A_2(x,t) := -
2ia\int_0^{\infty}du \int_0^\infty dy\, e^{-uva} \bigg[\frac{d}{dy}U_{t-t_1}(u+x+y)\bigg]e^{i\varphi(t_1)}e^{-i\frac{v}{2}y}\phi(y-x_0+vt_1)\\
\end{equation*}
\begin{equation*}
A_3(x,t) := -
2ia\int_0^{\infty}du \int_0^\infty dy\, e^{-uva} U_{t-t_1}(u+x+y)e^{i\varphi(t_1)}e^{-i\frac{v}{2}y}\bigg[\frac{d}{dy}\phi(y-x_0+vt_1)\bigg]\,.
\end{equation*}
We notice that
\begin{equation*}
\bigg\|
\int_0^{\infty}du\, e^{-uv a} U_{t-t_1}(u+\cdot)\bigg\|_{L^2(\RE^+)}
\leqslant
\bigg\|
\int_0^{\infty}du\, e^{-uva} U_{t-t_1}(u+\cdot)\bigg\|_{L^2(\RE)}
=
 \| \chi_+ e^{-va \cdot} \|_{L^2(\RE)}
  =
 \sqrt{ \frac{1}{2 v a}}.
\end{equation*}
Then the following estimate for the term $A_1$ holds true
\begin{equation}
\label{estA1}
\|A_1(t)\|_{L^2(\RE^+)}\leqslant 2a
\bigg\|
\int_0^{\infty}du\, e^{-uva} U_{t-t_1}(u+\cdot)\bigg\|_{L^2(\RE^+)} |\phi(-x_0+vt_1)|\leqslant C\,\frac{e^{-v^{1-\de}}}{v^{1/2}}\,.
\end{equation}
The term $A_3$ is estimated by
\begin{equation}
\label{estA3}
\begin{aligned}
\|A_3(t)\|_{L^2(\RE^+)}\leqslant &
2 a\int_0^{\infty}du \,e^{-uva}
\big\| \big[U_{t-t_1}\chi_+e^{-i\frac{v}{2}\cdot}\phi'(\cdot-x_0+vt_1)\big]\big(-(u+\cdot)\big)\big\|_{L^2(\RE^+)}\\
\leqslant &
2\big\|\phi'\big\|_{L^2(\RE)}
a\int_0^{\infty}du \,e^{-uva }
\leqslant \frac{C}{v}\,,
\end{aligned}
\end{equation}
where we used the equality $[U_t^+f](x)=[U_t \chi_+ f](-x)$.

\noindent
We compute finally the term $A_2$. By integration by parts
\begin{equation*}
\begin{aligned}
A_2(x,t)=&-
2ia\int_0^{\infty}du \int_0^\infty dy\, e^{-uva}
\bigg[\frac{d}{du}U_{t-t_1}(u+x+y)\bigg]e^{i\varphi(t_1)}e^{-i\frac{v}{2}y}\phi(y-x_0+vt_1)
\\
=&
2ia \phi^+_{t}(x)
-2iv a^2
\int_0^{\infty}e^{-uva}\big[U^+_{t-t_1}\phi_{x_0,-v}(t_1)\big](u+x) du\,,
\end{aligned}
\end{equation*}
where the function $ \phi^+_{t}$ was defined in equation
\eqref{phi+}. Using the last equality in the equation \eqref{misunderstood} we get
\begin{equation}
\label{night}
va\int_0^{\infty}e^{-uva}\big[U^+_{t-t_1}\phi_{x_0,-v}(t_1)\big](u+x) du=
\frac{2ia}{1+2ia} \phi^+_{t}(x)+
\frac{A_1(x,t)+A_3(x,t)}{1+2ia}\,.
\end{equation}
From the definition of $\Phi^{S,in}_t$ and using the last equality
with $a=\tilde\alpha/3$ in the formula for the integral kernel of
$e^{- i H_\de^{\tilde \al v} t}$, see equation \eqref{propdelta}, it follows that
\begin{equation}
\label{carnival}
e^{-iH_\de^{\tilde \alpha v}(t-t_1)}\Phi_{t_1}^{S,in}=
e^{-iH_\de^{\tilde \alpha v} (t-t_1)}
\begin{pmatrix}
\phisol{x_0,-v}(t_1)\\ \\
0\\ \\
0
\end{pmatrix}
=
\begin{pmatrix}
 \phi^-_{t}-\frac{1+2i\tilde\al}{3+2i\tilde\al}\phi^+_t\\ \\
\frac{2}{3+2i\tilde\al}\phi^+_t\\ \\
\frac{2}{3+2i\tilde\al}\phi^+_t
\end{pmatrix}
-\frac{2}{3}
\begin{pmatrix}
A_{\tilde\alpha}(t)\\ \\
A_{\tilde\alpha}(t)\\ \\
A_{\tilde\alpha}(t)
\end{pmatrix}\,,
\end{equation}
where the function $\phi_t^-$ was defined in equation \eqref{phi-} and
we set $A_{\tilde\alpha}(x,
t):=\big[(A_1(x,t)+A_3(x,t))/(1+2ia)\big]\big|_{a=\tilde\alpha/3}$. Similarly
using equality \eqref{night} with $a=3/\tilde \beta$ in the formula
for the integral kernel of the propagator $e^{- i H_{\de'}^{\tilde
    \beta / v} t}$, see equation \eqref{propdeltaprime}, we get
\begin{equation}
\label{carnival2}
e^{-iH_{\de'}^{\tilde \beta / v}(t-t_1)}\Phi_{t_1}^{S,in}=
\begin{pmatrix}
 \phi^-_{t}+\frac{\tilde \beta+2i}{\tilde\beta+6i}\phi^+_t\\ \\
- \f {4i} {\tilde \beta + 6 i}\phi^+_t\\ \\
- \f {4i} {\tilde \beta + 6 i}\phi^+_t
\end{pmatrix}
-\frac{2}{3}
\begin{pmatrix}
A_{\tilde\beta}(t)\\ \\
A_{\tilde\beta}(t)\\ \\
A_{\tilde\beta}(t)
\end{pmatrix}
\end{equation}
where we introduced the notation $A_{\tilde\beta}(x, t):=\big[(A_1(x,t)+A_3(x,t))/(1+2ia)\big]\big|_{a=3/\tilde\beta}$. We notice that the analogous formula for $e^{-iH_F(t-t_1)}\Phi_{t_1}^{S,in}$ can be obtained from equation \eqref{carnival} by setting $\tilde\alpha=0$ and $A_{\tilde\alpha}=0$.

To get the estimate \eqref{seesee} we show that, at the cost of an error of the order $(t_2-t_1)$, for $t\in(t_1,t_2)$, the functions $\phi^-_{t}(x)$ and $\phi^+_{t}(x)$ can be approximated by the solitons $\phisol{x_0,-v}(x,t)$ and $\phisol{-x_0,v}(x,t)$ respectively. We consider first the function $\phi^+_{t}$, by adding and subtracting a suitable term to the r.h.s. of equation \eqref{phi+} we get
\begin{equation}
\label{pippo1}
\begin{aligned}
\phi^+_{t}(x)=&
\int_{-\infty}^\infty U_{t-t_1}(x+y)\phisol{x_0,-v}(y,t_1)
+ i \int_{t_1}^{t} ds
\int_{-\infty}^\infty U_{t-s} (x + y)| \phisol{x_0,-v}(y,s)|^2 \phisol{x_0,-v}(y,s)
 dy
\\
&-
\int_{-\infty}^0 U_{t-t_1}(x+y)\phisol{x_0,-v}(y,t_1) dy - i \int_{t_1}^{t} ds
\int_{-\infty}^\infty U_{t-s} (x + y)|\phisol{x_0,-v}(y,s) |^2\phisol{x_0,-v}(y,s) \, dy\\
=&\phisol{-x_0,v}(x,t) +I(x,t)+II(x, t)\,,
\end{aligned}
\end{equation}
where we used the fact that $\phisol{x_0,-v}(-x,t)=\phisol{-x_0,v}(x,t)$ and we set
\begin{equation*}
I(x,t):=-
\int_{-\infty}^0 U_{t-t_1}(x+y)\phisol{x_0,-v}(y,t_1) dy
\end{equation*}
and
\begin{equation*}
II(x,t):=- i \int_{t_1}^{t} ds
\int_{-\infty}^\infty U_{t-s} (x + y)|\phisol{x_0,-v}(y,s) |^2\phisol{x_0,-v}(y,s) \, dy\,.
\end{equation*}
For the term $I$ we use the estimate
\begin{equation*}
\|I\|_{L^2(\RE^+)}\leqslant \|\phi (\cdot - x_0 + vt_1)\|_{L^2(\RE^-)}
\leqslant 2e^{-v^{1-\de}}\,.
\end{equation*}
The term $II$ is estimated by
\begin{equation*}
\|II\|_{L^2(\RE^+)}\leqslant (t-t_1) \|\phi^3\|_{L^2(\RE)}\leqslant C v^{-\de}\,.
\end{equation*}

\noindent
Similarly, for the function $\phi^-_{t_2}$, we get
\begin{equation}
\label{pippo2}
\phi^-_{t}(x)=\phisol{x_0,-v}(x,t)+III(x,t)+IV(x,t)
\end{equation}
where we set
\begin{equation*}
III(x,t):=
-
\int_{-\infty}^0 U_{t-t_1}(x-y)\phisol{x_0,-v}(t_1,y) dy
\end{equation*}
and
\begin{equation*}
IV(x,t):= - i \int_{t_1}^{t} ds
\int_{-\infty}^\infty U_{t-s} (x - y)|\phisol{x_0,-v}(s,y)|^2\phisol{x_0,-v}(s,y)
\, dy\,.
\end{equation*}
For $t\in(t_1,t_2)$ the estimates
\begin{equation*}
\|III\|_{L^2(\RE^+)}\leqslant 2e^{-v^{1-\de}}\;;\quad \|IV\|_{L^2(\RE^+)}\leqslant C v^{-\de}
\end{equation*}
are similar to the ones given above for the terms $I$ and $II$. Then from equations \eqref{pippo1} and \eqref{pippo2} we have
\begin{equation*}
\|\phi^+_{t}-\phisol{-x_0,v}(t)\|_{L^2(\RE^+)}\leqslant C[e^{-v^{1-\de}}+v^{-\de}]\,,\quad
\|\phi^-_{t}-\phisol{x_0,-v}(t)\|_{L^2(\RE^+)}\leqslant C[e^{-v^{1-\de}}+v^{-\de}]
\end{equation*}
for all $t\in(t_1,t_2)$. By using the last estimates and the estimates \eqref{estA1} and \eqref{estA3} in equations \eqref{carnival} and \eqref{carnival2} we get
\begin{equation*}
\big\|e^{-iH(t-t_1)}\Phi_{t_1}^{S,in}-\Phi^S_{t} \big\|\leqslant
 C\bigg[e^{-v^{1-\de}}+v^{-\de}+\frac{e^{-v^{1-\de}}}{v^{1/2}}+\frac{1}{v}\bigg]
\end{equation*}
which in turn implies that for $v$ big enough the estimate \eqref{seesee} holds true.
 
Using estimates \eqref{L2S1}, \eqref{L2S2} and \eqref{seesee} in the inequality \eqref{ineq1} we get the estimate \eqref{prop2.1}.

The estimate \eqref{prop2.2} is a consequence of estimate \eqref{prop2.1} and of the fact that $\|\Phi^S_{t_2}-\Phi^{S,out}_{t_2}\|=\|\Phi^{S,in}_{t_2}\|\leqslant C e^{-v^{1-\de}}$.
\end{proof}
\begin{rem}{\em
Notice that estimate \eqref{prop2.1}, although not strictly necessary
for the proof of theorem \ref{mainth}, enforces the picture that in
the phase 2 a scattering event is occurring. The true wavefunction can
be approximated by the superposition of an incoming and an outgoing
wavefunction. At the end of this phase, only the outgoing wavefunction
is not negligibile.}
\end{rem}

\subsection{Phase 3}
Let us put $t_3=t_2+T\ln v$. We call ``phase 3'' the evolution of the system in the time interval $(t_2,t_3)$.
The approximation of $\Psi_t$ during this time interval is the content of theorem \ref{mainth}.

\n
We recall the following result (see \cite{[HMZ07]}).
\begin{prop}
\label{appAHMZ}
Let $\phi^{tr}_t$ and $\phi^{ref}_t$ be defined as it was done in equations \eqref{phiref} and \eqref{phitr} above. Then $\forall k\in \NN$ there exist two constants $c(k)>0$ and $\si(k)>0$ such that
\begin{equation}
\|\phi^\gamma_t\|_{L^2(\RE^-)}+\|\phi^\gamma_t\|_{L^\infty(\RE^-)}\leqslant \frac{c(k)(\ln v)^{\si(k)}}{v^{k(1-\de)}}
\label{shodo}
\end{equation}
for $\gamma = \{ \text{ref} , \text{tr} \}$, uniformly in $t\in[0,T\ln v]$.
\end{prop}
Notice that also the norms $\|\phi^\gamma_t\|_{L^p(\RE^-)}$, for $2 \leqslant p \leqslant \infty$, are estimated by the r.h.s. of
\eqref{shodo}.
Now we can prove theorem \ref{mainth}
\begin{proof}[Proof of Theorem \ref{mainth}.]
The strategy of the proof closely follows proposition \ref{fase1}. We will just sketch the common part of the proof while
proving in details the different estimates.

Let us define $ \Xi_t:=\Psi_t-\sum_{j=1}^3\Phi^j_t$ where the vectors
$\Phi^j_t$ were given in equation \eqref{Phijt} and fix $t_a\in
[t_2,t_3]$. From equations \eqref{intform} and \eqref{Phijt} it
follows that the vector $\Xi_t$ satisfies the following integral
equation
\begin{equation}
\label{amaranto}
\begin{aligned}
\Xi_t=&
e^{-iH (t-t_a)} \Xi_{t_a} \\
&+\sum_{j=1}^3\bigg[\Big(e^{-iH (t-t_a)}-e^{-iH_j (t-t_a)}\Big) \Phi^j_{t_a} +
i \int_{t_a}^t ds \, \Big( e^{-iH (t-s)}- e^{-i H_j (t-s)}\Big) | \Phi^j_s |^2 \Phi^j_s\bigg]\\
&+ i \int_{t_a}^t ds \, e^{-i
  H (t-s)} \sum_{j_1,j_2,j_3}(1-\de_{j_1j_2}\de_{j_2j_3})\overline{\Phi^{j_1}_s} \Phi^{j_2}_s \Phi^{j_3}_s\\
&+ i \int_{t_a}^t ds \, e^{-i
  H (t-s)} \bigg[\bigg| \Xi_s +\sum_{j=1}^3 \Phi^j_s\bigg|^2 \Xi_s +|\Xi_s|^2\sum_{j=1}^3 \Phi^j_s+2\Re\bigg[\overline{ \Xi_s}\sum_{j=1}^3 \Phi^j_s \bigg]\sum_{j=1}^3 \Phi^j_s\bigg]\\
=& e^{-iH (t-t_a)} \Xi_{t_a} + G(t_a,t) \\
&+ i \int_{t_a}^t ds \, e^{-i H (t-s)} \bigg[\bigg| \Xi_s +\sum_{j=1}^3 \Phi^j_s\bigg|^2 \Xi_s +|\Xi_s|^2\sum_{j=1}^3 \Phi^j_s+2\Re\bigg[\overline{ \Xi_s}\sum_{j=1}^3 \Phi^j_s \bigg]\sum_{j=1}^3 \Phi^j_s\bigg]\,.
\end{aligned}
\end{equation}

\n
Let us fix $t_b\in [t_a,t_3]$ and let $X_{t_a,t_b}= L^{\infty}_{[t_a, t_b] } L^2\cap L^{6}_{[t_a, t_b] }L^6 $. Using the Strichartz estimates as it was done in proposition \ref{fase1} (see equations \eqref{split} - \eqref{split3}), it is straightforward to prove that
\begin{equation*}
\| \Xi \|_{X_{t_a,t_b}} \leqslant \wt C\lf[ \| \Xi_{t_a} \| + \| G(t_a,\cdot)\|_{X_{t_a,t_b}} + (t_b-t_a)^{1/2} \lf( \| \Xi \|_{X_{t_a,t_b}}^2 + \| \Xi \|_{X_{t_a,t_b}}^3 \ri) \ri]
\end{equation*}
where $\wt C$ depends only on the constants appearing in the Strichartz estimates. Using lemma \ref{blocco} as it was done in the proof of proposition \ref{fase1} it follows that there exists $\tau>0$ such that, for any $t \in [t_2, t_3)$ one has
$\| \Xi \|_{X_{t, t + \tau}} \ls \ 2 \wt C \ \| \Xi_{t} \| + 2 \wt C \ \| G(t, \cdot) \|_{X_{t, t+ \tau}}$.

\n
We divide the interval $[t_2,t_3]$ in $N+1$ subintervals:
$[t_2+j\tau,t_2+(j+1)\tau)$, with $ j=0,...,N-1$; and $[t_2+N\tau,
t_3)$, and where $N$ is the integer part of $(t_3-t_2)/\tau$.

\n
Then proceeding by induction as we did in the proof of proposition
\ref{fase1}, see equations \eqref{indu} and \eqref{almost}, we get the
inequality
\begin{equation}
\label{perf}
\| \Xi \|_{X_{t_2,t_3}} \ \ls \sum_{j=0}^{N} (2\wt C)^{j+1} \| \Xi_{t_2}
\| + \sum_{j=0}^{N} \sum_{k = 0}^{j} (2 \wt C)^{j+1-k} \|
G (t_2+k\tau, \cdot) \|_{X_{t_2+k\tau, \min\{t_3,t_2+ (k+1) \tau\}}} \,.
\end{equation}

Now we estimate the initial  data $\| \Xi_{t_2}\|$ and the
source term $\|G (t_a, \cdot) \|_{X_{t_a,t_b}}$ with
$t_2\leqslant t_a\leqslant t_b\leqslant t_3$ and $t_b-t_a\leqslant\tau$.

\n
By proposition \ref{lemma2} (estimate \eqref{estA3}) and using the definitions \eqref{Phi1t2} - \eqref{Phi3t2} one has
\begin{equation*}
\Xi_{t_2} =
 \Psi_{t_2}-\sum_{j=1}^3\Phi^j_{t_2}=\Psi_{t_2}-\Phi^{S,out}_{t_2}
-\begin{pmatrix}
 \tilde t\,\phi_{x_0,-v}(t_2)\\ \\
\tilde r\,\phi_{x_0,-v}(t_2)\\ \\
\tilde t\,\phi_{x_0,-v}(t_2)
\end{pmatrix}
\,,
\end{equation*}
with $\| \Psi_{t_2}-\Phi^{S,out}_{t_2}\|\leqslant C v^{-\de/2}$. Since
\begin{equation*}
\int_0^\infty|\phi_{x_0,-v} (x,t_2)|^2dx=\int_0^\infty|\phi (x - x_0 + vt_2)|^2dx=
2\int_{v^{1-\de}}^\infty\sech(x)^2dx=\frac{4 e^{-2v^{1-\de}}}{1+e^{-2v^{1-\de}}}\leqslant 4 e^{-2v^{1-\de}}
\end{equation*}
we have
\begin{equation}
\label{ect}
\|\Xi_{t_2} \|\leqslant C\big(v^{-\f \de 2}+e^{-2v^{1-\de}}\big)\,\leqslant C v^{-\f \de 2} \,.
\end{equation}
Let us now consider the source term $G(t_a,t)$. We use the estimate $\|G(t_a,\cdot)\|_{X_{t_a,t_b}}\leqslant \|G(t_2,\cdot)\|_{X_{t_2,t_3}}$. To simplify the notation we set $G(t) \equiv G(t_2,t)$ and
\begin{equation*}
G_1 (t) : =\sum_{j=1}^3\bigg[\Big(e^{-iH(t-t_2)}-e^{-iH_j (t-t_2)}\Big) \Phi^j_{t_2} +
i \int_{t_2}^t ds \, \Big( e^{-iH(t-s)}- e^{-i H_j (t-s)}\Big) | \Phi^j_s |^2 \Phi^j_s\bigg]
\end{equation*}
\begin{equation*}
G_2(t) : = i \int_{t_2}^t ds \, e^{-i
  H (t-s)} \sum_{j_1,j_2,j_3}(1-\de_{j_1j_2}\de_{j_2j_3})\overline{ \Phi^{j_1}_s} \Phi^{j_2}_s \Phi^{j_3}_s\,.
\end{equation*}
By the definition of $G(t_a,t)$, see equation \eqref{amaranto} it follows that
\begin{equation*}
G(t)=G_1(t)+G_2(t)\,;
\end{equation*}
we estimate $G_1(t)$ and $G_2(t)$ separately.

We proceed first with the estimate of the term $G_1$. From equations \eqref{freepro}, \eqref{propdelta}, \eqref{propdeltaprime} and \eqref{twoedgepro-j} one can see that for any (column) vector $F=(F_1,F_2,F_3)\in L^2$
\begin{equation*}
\big[e^{-iH t}-e^{-iH_j t}\big]F= \mathbb{M}_j
\begin{pmatrix}
U^+_tF_1\\ \\
U^+_tF_2\\ \\
U^+_tF_3
\end{pmatrix}
-\frac{2}{3}\mathbb{J}
\begin{pmatrix}
v a \int_0^\infty e^{-u v a} \big[U^+_tF_1\big](u+\cdot) du\\ \\
v a \int_0^\infty e^{-u v a} \big[U^+_tF_2\big](u+\cdot) du\\ \\
v a \int_0^\infty e^{-u v a} \big[U^+_tF_3\big](u+\cdot) du
\end{pmatrix}
\qquad j=1,2,3\,,
\end{equation*}
where the constant $a$ and the matrices $ \mathbb{M}_j$ must be chosen
accordingly to the Hamiltonian $H$:\\
for $H=H_\de^{\tilde \alpha v}$,
\begin{equation*}
a=\frac{\tilde \alpha}{3}\;,\qquad\mathbb{M}_j=-\mathbb{I}+\frac{2}{3}\mathbb{J}-\mathbb{T}_j\,;
\end{equation*}
for $H=H_{\de'}^{\tilde \beta / v}$,
\begin{equation*}
a=\frac{3}{\tilde \beta}\;,\qquad
\mathbb{M}_j=\mathbb{I}-\mathbb{T}_j\,;
\end{equation*}
and the formula for $H=H_F$ can be obtained by setting $\tilde\alpha=0$ in the formula for $H=H_\de^\alpha$.

\n
Then, denoting by $\big(G_1(x,t)\big)_l$, $l=1,2,3$, the $l$-th component of the vector $G_1$ one has
\begin{equation*}
\big(G_1(x,t)\big)_l=\big(\widetilde{G}_1(x,t)\big)_l+\big(\widehat{G}_1(x,t)\big)_l\,,
\end{equation*}
with
\begin{equation*}
\big(\widetilde{G}_1(x,t)\big)_l:=
\sum_{j,k=1}^3\big(\mathbb{M}_j\big)_{lk}\bigg[\big[U^+_{t-t_2}\big(\Phi^j_{t_2}\big)_k\big](x)+
i\int_{t_2}^tds\big[U^+_{t-s}\big|\big(\Phi^j_{s}\big)_k\big|^2\big(\Phi^j_{s}\big)_k\big](x)\bigg]
\end{equation*}
\begin{equation*}
\big(\widehat{G}_1(x,t)\big)_l:=
-\frac{2va}{3}\sum_{j,k=1}^3\int_0^\infty e^{-vua}\bigg[\big[U^+_{t-t_2}\big(\Phi^j_{t_2}\big)_k\big](u+x)+
i\int_{t_2}^tds\big[U^+_{t-s}\big|\big(\Phi^j_{s}\big)_k\big|^2\big(\Phi^j_{s}\big)_k\big](u+x)\bigg]du\,.
\end{equation*}
From the definition of the vectors $\Phi^j_{t}$, see equations
\eqref{Phi123t} - \eqref{Phi1234t},
 we see that for each $l=1,2,3$ the function $\big(\widetilde{G}_1(x,t)\big)_l$ is a linear combination of four functions, $f^{\gamma,+}_t$ and $f^{\gamma,-}_t$, with $\gamma $ being equal to $ref$ and $tr$, given by
\begin{equation}
\label{fal+}
f^{\gamma,+}_t(x) := e^{-i \f {v^2} 4 t_2}e^{it_2}\bigg[\int_0^\infty U_{t-t_2}(x+y)\phi_0^\gamma(y)dy+i\int_{t_2}^tds\int_0^\infty U_{t-s}(x+y)|\phi^{\gamma}_{s-t_2}(y)|^2\phi^{\gamma}_{s-t_2}(y)dy\bigg]
\end{equation}
and
\begin{equation*}
f^{\gamma,-}_t(x) := e^{-i \f {v^2} 4 t_2}e^{it_2}\bigg[\int_0^\infty U_{t-t_2}(x+y) \phi_0^\gamma(-y)dy+i\int_{t_2}^tds\int_0^\infty U_{t-s}(x+y)|\phi^{\gamma}_{s-t_2}(-y)|^2\phi^{\gamma}_{s-t_2}(-y)dy\bigg]\,,
\end{equation*}
where the functions $\phi^{ref}_t$ and $\phi^{ref}_t$ were defined in equations \eqref{phiref} and \eqref{phitr} respectively.

Similarly one can see that for each $l=1,2,3$ the function $\big(\widehat{G}_1(x,t)\big)_l$ is a linear combination of
\begin{equation*}
va \int_0^\infty e^{-vua} f^{ref,+}_t(u+x)du\;,\quad
va \int_0^\infty e^{-vua} f^{tr,+}_t(u+x)du\;,
\end{equation*}
\begin{equation*}
va \int_0^\infty e^{-vua} f^{ref,-}_t(u+x)du\;,\quad
va \int_0^\infty e^{-vua} f^{tr,-}_t(u+x)du\;.
\end{equation*}
First we study the function $f^{\gamma,+}_t$. We notice that, adding and subtracting a suitable term in equation \eqref{fal+} and using the definitions \eqref{phiref} and \eqref{phitr}, $f^{\gamma,+}_t(x)$ can be written as
\begin{equation*}
f^{\gamma,+}_t(x)
= I(x,t) + II(x,t) + III(x,t)\,,
\end{equation*}
with
\begin{equation*}
 I(x,t):= e^{-i \f {v^2} 4 t_2}e^{it_2}\phi^{\gamma}_{t-t_2}(-x)
\end{equation*}
\begin{equation*}
 II(x,t):= -e^{-i \f {v^2} 4 t_2}e^{it_2}\int_0^{\infty} U_{t-t_2}(x-y)\phi_0^\al(-y)dy
 \end{equation*}
\begin{equation*}
 III(x,t):=-ie^{-i \f {v^2} 4 t_2}e^{it_2}\int_{t_2}^{t}ds\int_{0}^\infty U_{t-s}(x-y)|\phi^{\gamma}_{s-t_2}(-y)|^2\phi^{\gamma}_{s-t_2}(-y)dy\,.
\end{equation*}
Similarly to what was done above,
we set $X_{t_2,t_3}(\RE^\pm)= L^{\infty}_{[t_2, t_3] } L^2 (\RE^\pm)
\cap L^{6}_{[t_2, t_3] } L^6 (\RE^\pm) $. By proposition \ref{appAHMZ},
we have
\begin{equation}
\label{glad0}
\| I \|_{X_{t_2,t_3}(\RE^+)}=\|\phi^{\gamma}_{\cdot-t_2} \|_{X_{t_2,t_3}(\RE^-)}
\leqslant \frac{c'(k)(\ln v)^{\si '(k)}}{v^{k(1-\de)}}
\end{equation}
where $c'(k)$ and $\si '(k)$ are constants, different from the one
appearing in proposition \ref{appAHMZ}. For our purposes we do not
need to compute them.

\n
Using the one dimensional Strichartz estimates for $U_t$, we have
\begin{equation}
\label{glad1}
\| II\|_{X_{t_2,t_3}(\RE^+)}\leqslant C
\| \chi_- \phi_0^\gamma\|_{L^2(\erre) }
\leqslant C e^{-2v^{1-\de}}\,.
\end{equation}
Finally, the term ${III}$ can be estimated using the inhomogeneous Strichartz estimate and proposition \ref{appAHMZ}
\begin{equation}
\label{glad2}
\|III\|_{X_{t_2,t_3}(\RE^+)}
\leqslant C \|( \chi_- \phi^{\gamma}_{\cdot-t_2} )^3 \|_{L^1_{[t_2, t_3] } L^2(\erre)}
\leqslant C
\bigg[ \frac{c(k)(\ln v)^{\si(k)'}}{v^{k(1-\de)}}\bigg]^3\,.
\end{equation}
Collecting the estimates \eqref{glad0}, \eqref{glad1} and \eqref{glad2}, it follows that
\begin{equation*}
\|f^{\gamma,+}_\cdot\|_{X_{t_2,t_3}(\RE^+)}\leqslant C
\bigg[ \frac{c(k)(\ln v)^{\si(k)'}}{v^{k(1-\de)}}\bigg]\,.
\end{equation*}
The estimate of $f^{\gamma,-}_t$ is similar and we omit it. We have proved
that for some $c'(k)$ and $\si'(k)$ possibly bigger than $c(k)$ and $\si(k)$ we have:
\begin{equation*}
\|\widetilde G_1\|_{X_{t_2,t_3}} \leqslant C
\bigg[ \frac{c'(k)(\ln v)^{\si'(k)}}{v^{k(1-\de)}}\bigg]\,,
\end{equation*}
where $\widetilde G_1$ is the vector in $L^2$ with components $\big(\widetilde{G}_1(x,t)\big)_l$, $l=1,2,3$.

The estimate of $\widehat G_1(t)=\big(\big(\widehat{G}_1(t)\big)_1,\big(\widehat{G}_1(t)\big)_2,\big(\widehat{G}_1(t)\big)_3\big)$ is a trivial consequence of the fact that
\begin{equation*}
\bigg\|va \int_0^\infty e^{-vua} f^{\gamma,\pm}_\cdot(u+\,\cdot\,)du\bigg\|_{X_{t_2,t_3}(\RE^+)}
\leqslant
\big\|f^{\gamma,\pm}_\cdot\big\|_{X_{t_2,t_3}(\RE^+)}\leqslant
C
\bigg[ \frac{c(k)(\ln v)^{\si(k)'}}{v^{k(1-\de)}}\bigg]\,,
\end{equation*}
from which it follows that
\begin{equation*}
\|\widehat G_1\|_{X_{t_2,t_3}} \leqslant C
\bigg[ \frac{c'(k)(\ln v)^{\si'(k)}}{v^{k(1-\de)}}\bigg]\,;
\end{equation*}
and
\begin{equation}
\label{estG1}
\| G_1\|_{X_{t_2,t_3}} \leqslant C
\bigg[ \frac{c'(k)(\ln v)^{\si'(k)}}{v^{k(1-\de)}}\bigg]\,.
\end{equation}

\n
We analyse now the term $G_2$. Due to the presence of $(1-\de_{j_1j_2}\de_{j_2j_3})$ the components of the vector
\begin{equation*}
(1-\de_{j_1j_2}\de_{j_2j_3})\overline{ \Phi^{j_1}_s}\Phi^{j_2}_s\Phi^{j_3}_s
\end{equation*}
contains only terms (up to a phase) like
\begin{equation*}
 \phi^{\gamma_1}_{t-t_2} (x) \phi^{\gamma_2}_{t-t_2} (x) \phi^{\gamma_3}_{t-t_2} (-x)\quad\textrm{or}\quad
 \phi^{\gamma_1}_{t-t_2} (x) \phi^{\gamma_2}_{t-t_2} (-x) \phi^{\gamma_3}_{t-t_2} (-x)
\end{equation*}
where $\gamma_1$, $\gamma_2$ and $\gamma_3$ can be $ref$ or $tr$. This can easily be seen by using equations \eqref{Phi123t} - \eqref{Phi1234t}.
By Strichartz methods, it is sufficient to estimate the $L^1_{[t_2, t_3] } L^2(\erre^+)$ norm of these terms. Then using H\"older's inequality we have, for istance
\begin{equation*}
\|\phi^{\gamma_1}_{t-t_2}\phi^{\gamma_2}_{t-t_2}\phi^{\gamma_3}_{t-t_2}(-\cdot)\|_{L^{2}(\RE^+)}\leqslant
\|\phi^{\gamma_1}_{t-t_2}\|_{L^4(\RE^+)}\|\phi^{\gamma_2}_{t-t_2}\|_{L^4(\RE^+)}\|\phi^{\gamma_3}_{t-t_2}\|_{L^{\infty}(\RE^-)}
\leqslant C \frac{c(k)(\ln v)^{\si(k)}}{v^{k(1-\de)}}\,.
\end{equation*}
The second kind of terms can be estimated in the same way and we obtain
\begin{equation*}
\|G_2\|_{X_{t_2,t_3}} \leqslant C
\bigg[ \frac{c'(k)(\ln v)^{\si'(k)}}{v^{k(1-\de)}}\bigg]\,,
\end{equation*}
which, together with the estimate \eqref{estG1}, gives
\begin{equation}
\|G\|_{X_{t_2,t_3}} \leqslant C
\bigg[ \frac{c'(k)(\ln v)^{\si'(k)}}{v^{k(1-\de)}}\bigg]\,.
\label{kagu}
\end{equation}
Fix $k$ such that $k(1-\de)>2$ then for $v$ sufficiently large \eqref{kagu} implies
\begin{equation}
\|G\|_{X_{t_2,t_3}} \leqslant \f 1 v\,.
\label{unten}
\end{equation}

\begin{equation*}
\begin{aligned}
 \| \Xi \|_{X_{t_2,t_3}} \ \ls & Cv^{-\de/2} \sum_{j=0}^{N} (2\wt C)^{j+1} + v^{-1} \sum_{j=0}^{N} \sum_{k = 0}^{j} (2 \wt C)^{j+1-k}\\
\leqslant &
2 Cv^{-\de/2} (2\wt C)^{N+1} + 2 v^{-1} \sum_{j=0}^{N} (2 \wt C)^{j+1}
\\
\leqslant &
2 Cv^{-\de/2} (2\wt C)^{N+1} + 2 v^{-1} (2 \wt C)^{N+1}\leqslant \widehat C v^{-\de/2}(2 \wt C)^N\,.
\end{aligned}
\end{equation*}
Since $N$ is the integer part of $(t_3-t_2)/\tau=\frac{T}{\tau}\ln v$ we have
\begin{equation*}
 \| \Xi \|_{X_{t_2,t_3}} \ \ls \ \widehat C v^{- \delta/2 + \f T \tau
   \ln (2 \wt C)}
\end{equation*}
We can finally set $\tau_*\equiv \tau/\ln (2 \wt C)$ and $T_*=\de \tau_*/2$ and obtain
\begin{equation*}
 \| \Xi \|_{L^\infty_{[t_2,t_3]}L^2}\leqslant \| \Xi \|_{X_{t_2,t_3}} \leqslant \widehat C v^{ -\frac{T_*-T}{\tau_*}}\,,
\end{equation*}
which concludes the proof of theorem \eqref{mainth}.
\end{proof}

\vskip5pt

\vspace{20pt}
\n

\section{Conclusion and perspectives}

In the present paper we have given a first rigorous analysis of
nonlinear Schr\"odinger propagation on graphs. We have given a
preliminary proof of local and global well posedness of the dynamics,
and of energy and mass conservation laws for some distinguished vertex
couplings, i.e Kirchhoff, $\delta$ and $\delta'$ couplings. Then we
concentrated on the problem of collision of a fast solitary wave on
the graph vertex (with couplings as before). It turns out that the
solitary wave splits in reflected and transmitted components the form
of which are again of solitary type, but with modified amplitudes
controlled by scattering coefficient given by the linear graph
dynamics. This behaviour holds true over times of the order $\ln v$
where $v$ is the velocity of the impinging soliton.

We add some other remarks on the result and further analysis and
generalizations.\p To begin with, let us note that the real line with
a point interaction at $0$ can be interpreted as a degenerate graph
with two edges. The cited paper \cite{[HMZ07]} treats the special case
of a $\delta$ interaction on the line, and our description shows how
it could be possible to extend their results to other point
interactions; among the examples treated in the present paper there is a version
of the $\delta'$ interaction,
showing how to treat point interactions of a more singular character than the one given by a $\delta$.\par
Concerning more general issues, a sharper
description of the post interaction phase can be achieved by an
explicit characterization of the evolution of the modified solitary
profiles, i.e. of the $ \Phi^j_t$ . This last part is somewhat
delicate, and intersects with contemporary intense work on asymptotics
for solitons in integrable and quasi integrable PDE, so we limit
ourselves to the following remarks. In the case analysed in
\cite{[HMZ07]}, the asymptotic behaviour of nonlinear Schr\"odinger
evolution of solitary waveforms with modified amplitudes is given, and
making use of inverse scattering theory it is shown (appendix B of the
cited paper) that the evolution is close to a soliton up to times of
order $\ln v$ and an error the order of which is an inverse power of
$v$. Borrowing from these results, it is possible to get in our case
too, but we omit details, the asymptotics of the $ \Phi^j_t$, i.e. of
the free NLS evolution of the modified solitary profiles outgoing from
the phase two. It turns out that these outgoing wavefunctions can be
approximated, on the same logarithmic timescale of Theorem 1.1, as new
solitons with the same waveform of the unperturbed dynamics, modified
amplitudes and phases, plus a dispersive (``radiation")
contribution. The meaning of this statement is that the $L^{\infty}$
norm of the difference between the evolved modified solitary profiles
$ \Phi^j_t$ and such final outgoing solitons has the usual dispersive
behaviour, $|t-t_2|^{-\frac{1}{2}}$. Let us note that, following this strategy, at the end of
the phase three, there would be two types of errors: errors due to the
approximation procedure in phase one and two (${\mathcal{O}}_{L^2}$); and errors arising from neglecting dispersion in the reconstruction
of the outgoing solitons (${\mathcal{O}}_{L^{\infty}}$) .\par An
important question concerns the possibility of extending the timescale
of validity of approximation by the solitary outgoing waves. As a
quite generic remark, this possibility could be related to the
asymptotic stability of the system, or of systems immediately related
to it.\p More concretely, in a different type of model (scattering of
two solitons on the line) in the already cited paper \cite{[AbFS]},
some considerations are given on obtaining longer timescales of
quasiparticle approximation in dependence of the initial data and
external potential, but it is unclear whether similar
considerations can be applied to the present case.
\p
Another issue is the nonlinearity. The fundamental asymptotics proved
in \cite{[HMZ07]} and used in the present paper relies on the
integrable nature of cubic NLS, and it is not immediate to extend
these results to more general nonlinearities. One can conjecture that
for nonlinearities close to integrable ones which admit solitary
waves, the outgoing waves are close to solitons over suitable
timescales. Let us mention, however, the recent results of Perelman on
the asymptotics of colliding solitons for nonlinearity close to
integrable or $L^2$ critical on the line (\cite {[P1]} and \cite{[P2]}
). \par A final problem is the extension of results
of the present work to more general graphs. We believe that results
similar to the ones of the present papers are valid for more general
boundary conditions at the vertex of a star graphs, with the same
proof, under the condition of absence of eigenvalues for the linear
Hamiltonian describing the graph. 
In presence of eigenvalues, some
Strichartz estimates weaken, and a more refined analysis is needed
(see \cite{[DH]} for the analogous problem on the line with an
attractive $\delta$ interaction). \p Of course, the extension of the present results to the case of star graphs with more than three edges has to be considered straightforward, while the extension to
graphs having a less trivial topology is an open problem.\p Finally
let us comment briefly the recent paper \cite{[Sob]}. In this partly
heuristic paper the authors study a star graph (but also more general type of graphs are
considered) with a NLS in which on every edge there is a different strenght $\beta_k$  in front of the the cubic term.  The authors fix a boundary condition which guarantees that mass and energy of the solution are (formally) conserved. Moreover according to the authors it is possible to derive a condition on the strenghts ${\beta}_k$ which allow for complete transmission of an incoming
solitary wave across the vertex. In these same situations the authors
show that an infinite chain of conserved quantities exists, defined
analogously to the case of the NLS on the line. The result, if formal,
is interesting, and concerning the relation with ours we note the
following.  In the case of a three edge graph and more generally for a
odd edge number, the complete transmission is made possible exactly by
the fine tuning of the coupling constants in front of the
nonlinearities. For the case of a single medium with the same
nonlinearity on every edge and Kirchhoff boundary conditions, one can
prove (see \cite{[ACFN2]}, where more generally the case of nonlinear
bound states for $\delta$ boundary conditions is treated) that exact
travelling solitons exist only in the case of a graph with an even
number of edges, while in the case of an odd number of edges a
stationary state is formed which is given by half a free soliton on
every edge.

\vskip5pt {\bf Acknowledgements.}
 
\n The present research was partially supported by INDAM-GNFM research
project ``Equazione di Schr\"odinger non lineare interagente con
difetti sulla retta e su grafi''. The Hausdorff Research Institute for
Mathematics is also acknowledged for the support.  The authors are grateful to Sergio
Albeverio for comments and discussions.





\end{document}